\documentclass[11pt,a4wide,runningheads]{article}

\usepackage{geometry}

\usepackage{amsfonts}
\usepackage{amsthm}

\usepackage{cite}

\usepackage{subfigure,diagbox}

\usepackage{tikz}
\usetikzlibrary{positioning}
\usetikzlibrary{calc}

\usepackage{comment}
\usepackage{amsmath}
\usepackage{complexity}

\usepackage{xcolor}
\usepackage{soul,color}
\usepackage{graphicx}

\usepackage[shortlabels]{enumitem}

\usepackage[affil-sl]{authblk} 

\usepackage{hyperref,xspace}

\newcommand{\dep}{{\sf dep}\xspace}
\renewcommand{\FPT}{{\sf FPT}\xspace}
\newcommand{\ETH}{{\sf ETH}\xspace}

\definecolor{dark-red}{rgb}{0.4,0.15,0.15}
\definecolor{dark-blue}{rgb}{0.15,0.15,0.4}
\definecolor{medium-blue}{rgb}{0,0,0.5}
\definecolor{gray}{rgb}{0.5,0.5,0.5}
\definecolor{color-Ig}{rgb}{0.15,0.7,0.15}

\hypersetup{
colorlinks = true,
urlcolor={myblue}, 
linkcolor = black!30!blue,
citecolor = black!30!green
}

\newcommand{\yes}{{\sf yes}\xspace}

\newtheorem{theorem}{Theorem}
\newtheorem{corollary}{Corollary}
\newtheorem{lemma}{Lemma}
\newtheorem{question}{Question}

\title{A Unifying Model for Locally Constrained\\Spanning Tree Problems}
\author[1]{Luiz Alberto do Carmo Viana}
\author[2]{Manoel Camp\^elo}
\author[3]{Ignasi Sau}
\author[4]{Ana Silva}

\affil[1]{\emph{\small Campus de Crate\'us, Universidade Federal do Cear\'a, Crate\'us, Brazil}}
\affil[2]{\emph{\small Dep. de Estatística e Matemática Aplicada, Universidade Federal do Cear\'a, Fortaleza, Brazil}}
\affil[3]{\emph{\small LIRMM, Universit\'e de Montpellier, CNRS, Montpellier, France}}
\affil[4]{\emph{\small Departamento de Matem\'atica, Universidade Federal do Cear\'a, Fortaleza, Brazil\newline \texttt{luizalberto@crateus.ufc.br, mcampelo@lia.ufc.br, ignasi.sau@lirmm.fr, anasilva@mat.ufc.br}}}

\begin{document}

\maketitle

\begin{abstract}
  Given a graph $G$ and a digraph $D$ whose vertices are the edges of
  $G$, we investigate the problem of finding a spanning tree of $G$
  that satisfies the constraints imposed by $D$. The restrictions to
  add an edge in the tree depend on its neighborhood in $D$. Here, we
  generalize previously investigated problems by also considering as
  input functions $\ell$ and $u$ on $E(G)$ that give a lower and an
  upper bound, respectively, on the number of constraints that must be
  satisfied by each edge. The produced feasibility problem is denoted
  by \texttt{G-DCST}, while the optimization problem is denoted by
  \texttt{G-DCMST}.  We show that \texttt{G-DCST} is $\NP$-complete
  even under strong assumptions on the structures of $G$ and $D$, as
  well as on functions $\ell$ and $u$.  On the positive side, we prove
  two polynomial results, one for \texttt{G-DCST} and another for
  \texttt{G-DCMST}, and also give a simple exponential-time algorithm
  along with a proof that it is asymptotically optimal
  under the \ETH. Finally, we prove that other previously studied
  constrained spanning tree (\textsc{CST}) problems can be modeled within our
  framework, namely, the \textsc{Conflict CST}, the \textsc{Forcing
    CST}, the \textsc{At Least One/All Dependency CST}, the
  \textsc{Maximum Degree CST}, the \textsc{Minimum Degree CST}, and
  the \textsc{Fixed-Leaves Minimum Degree CST}.
\end{abstract}

\section{Introduction}

Let $G$ be a graph and $D$ be a (directed or undirected) graph whose
vertices are the edges of $G$. In other terms, $D$ defines a relation
on the edge set of $G$.  The \emph{dependencies} of an edge
$e\in E(G)$ are given by its (in-)neighborhood in $D$, i.e., by the
set $\dep_D(e) = \{e'\in E(G)\mid (e',e)\in E(D)\}$. We omit the subscript $D$ in
$\dep_D(e)$ whenever the dependency graph is clear from the context.

Many problems have been investigated under the light of dependencies
between pairs of objects, such as the knapsack problem~\cite{GR.19},
bin-packing~\cite{EFL.11}, maximum flow~\cite{PS.13}, scheduling
problems~\cite{BJ.93}, maximum matchings~\cite{darmann2011paths}, shortest
paths~\cite{darmann2011paths}, or maximum acyclic
subgraphs~\cite{MapaUrrutia2015}. Generally, the dependency problems
defined on graphs can be described as the problem of finding a
subgraph $H$ of $G$ satisfying the dependency constraints imposed by
$D$.

However, the notion of dependency itself may vary. For example, every
$(e, e')$ in $D$ could mean that, whenever $e'$ is chosen (not chosen), we
get that $e$ cannot (must) be chosen, thus expressing a \emph{conflict
  constraint} (\emph{forcing constraint}),  
  with $D$ being the \emph{conflict graph} (\emph{forcing graph}).
In this paper, we introduce a generalization of dependency constrained
problems, and investigate this generalization for spanning trees. In particular, our model generalizes many of the constrained
spanning tree problems that have been investigated in the literature.

\paragraph{Our contribution.} For the generalized version of
dependency constrained problems, together with graph $G$ and (di)graph
$D$, we also consider functions $\ell$ and $u$ that assign, to each
$e\in E(G)$, a lower and an upper bound on the number of dependencies
that must be ensured for $e$. This means that a subgraph
$H\subseteq G$ satisfies the imposed constraints if and only if the
number of edges in $E(H)\cap \dep(e)$ is at least $\ell(e)$ and at
most $u(e)$, for every $e\in E(G)$\footnote{We can always assume that
  $0\leq \ell(e)\leq u(e)\leq |\dep(e)|$, and so $\ell(e)=u(e)=0$ if
  $\dep(e)=\emptyset$.}; we then say that $H$
\emph{$(\ell,u)$-satisfies $D$}. When $H$ is asked to be a spanning
tree, we call the problem the \textsc{Generalized Dependency
  Constrained Spanning Tree} problem, and denote it by
\texttt{G-DCST}. Also, sometimes we deal with the related optimization
problem by considering weights on the edges of $G$;
this is called the \textsc{Generalized Dependency Constrained Minimum
  Spanning Tree} problem, and is denoted by \texttt{G-DCMST}. Let us
observe that when $\ell(e)=0$ and $u(e)\geq |\dep(e)|$, for all
$e\in E(G)$, then \texttt{G-DCST} is equivalent to deciding whether
$G$ is connected, and \texttt{G-DCMST} corresponds to the classical
\textsc{Minimum Spanning Tree} problem.

Clearly, the feasibility problem \texttt{G-DCST} is a particular case
of the optimization problem \texttt{G-DCMST}, where the weight of each
edge is equal to one. This is why, whenever possible, we give preference to prove
$\NP$-completeness results for the feasibility problem, and get
polynomial results for the optimization problem.
We use reductions from $(3, 2, 2)$\texttt{-SAT} to prove our main $\NP$-completeness results, whereas our polynomial results arise as consequences of the Matroid Intersection Theorem~\cite{edmonds1979matroid} (cf. Section~\ref{sec:defs}).

Considering the generalized version of the spanning tree problem,
given a graph $G$, a digraph $D = (E(G),A)$, and functions $\ell,u$,
we prove that deciding whether $G$ has a spanning tree that
$(\ell,u)$-satisfies $D$ is $\NP$-complete in the following cases:
\begin{enumerate}[(i)]
\item\label{NP1} $\ell(e) = u(e) = |\dep(e)|$ for every $e\in E(G)$,
  $D$ is a forest of oriented paths of length at most two where all
  components are directed paths, out-stars, or in-stars, and $G$ is an
  outerplanar chordal graph with diameter at most~two. Furthermore, this
  problem cannot be solved in time $2^{o(n+m)}$ unless the \ETH fails,
  where $n=|V(G)|$ and $m=|E(G)|$;
\item\label{NP2} When $\ell,u$ are constant functions, for every pair
  of constant values such that  $\ell\le u$;
\item\label{NP3} When $\ell(e)=0$, and $u(e)=|\dep(e)|-1$, for every
  $e\in E(G)$.
\end{enumerate}

On the positive side, we prove the following:
\begin{enumerate}[(a)]
\item\label{P1} \texttt{G-DCST} can be solved in polynomial time when
  $D$ is an oriented matching, and
  $\ell(e) = u(e) = |\dep(e)|$ for every $e\in E(G)$;
\item\label{P2} \texttt{G-DCMST} can be solved in polynomial time when
  $\ell=0$, $D$ is a collection of symmetric complete digraphs
  $D_1,\ldots,D_k$, and $u(e)=u(e')$ whenever $e,e'$ are within the
  same component of $D$;
\item\label{P3} \texttt{G-DCST} can be solved in time
  $O(2^{m}\cdot (n+m))$, where $n=|V(G)|$ and $m=|E(G)|$.
\end{enumerate}

It is worth observing that~\ref{P1} and \ref{NP1} define a dichotomy between polynomial and hard cases for \texttt{G-DCST} when regarding $D$ as a family of oriented paths. It can be solved in polynomial time if length of the longest path in the underlying graph of $D$ is at most one, and it is \NP-Complete otherwise.

We also prove that many of the constrained spanning tree (CST)
problems that have been investigated in the literature can be modeled
with our general problem, namely the \textsc{Conflict
  CST}~\cite{DPS.09}, the \textsc{Forcing
  CST}~\cite{darmann2011paths}, the \textsc{At Least One/All
  Dependency CST}~\cite{vianaCampelo2019}, the \textsc{Maximum Degree
  CST}~\cite{deo1968shortest}, the \textsc{Minimum Degree
  CST}~\cite{de2012min}, and the \textsc{Fixed-Leaves Minimum Degree
  CST}~\cite{dias2017min}. All our reductions preserve the value of
the solutions, which means that also the optimization version of these
problems can be modeled within our framework.


Notice that the previously mentioned CST problems impose (vertex-wise
or edge-wise) local constraints to describe their set of feasible
spanning trees.  This contrasts with \textsc{Maximum Diameter
  CST}\cite{camerini1980complexity,camerini1983complexity}, \textsc{Minimum Diameter CST} \cite{hassin1995minimum,ho1991minimum} (with variations \cite{konemann2005approximating,BasteGPSST17}) and
\textsc{Maximum Leaves CST}~\cite{garey1979computers,lu1992power}, examples of
$\NP$-hard problems that impose constraints on global tree parameters.
In~\cite{DK.97}, the authors propose an approach that includes also
these global constraints, but from a practical point of view.

\paragraph{Related work.} In what follows, we talk sometimes about the
feasibility version of the problems, and sometimes about the
optimization version, where also a weight function on the edges of the
input graph is given. Also, when $\ell$ and/or $u$ are constant
functions, we write directly the constant value inside the parenthesis
when saying whether a spanning tree $(\ell,u)$-satisfies $D$.

\paragraph{{\it Conflict constraints:}}
Recall that, in the \textsc{Conflict Constrained (Minimum) Spanning
  Tree} problem, we are given a pair of graphs $G$ and $D$ such that
$V(D) = E(G)$, and we want to know whether there exists a spanning
tree (find a minimum spanning tree) $T$ of $G$ such that
$E(T)\cap \dep(e)=\emptyset$ for every $e\in E(T)$. We denote the
feasibility problem by \texttt{CCST} and the optimization problem by
\texttt{CCMST}.  Note that, if we consider $D'$ as an arbitrary orientation of $D$ (i.e., each edge
$e_1e_2$ in $D$ gives rise to either $(e_1, e_2)$ or $(e_2, e_1)$ in
$D'$), then we get that such a tree exists if and only if there exists
a spanning tree $T$ that $(0,0)$-satisfies $D'$. This means that our
problem generalizes this one and therefore inherits the $\NP$-complete
results, as well as might help with some polynomial cases. Also,
observe that the problems related to results~\ref{NP2} when
$\ell=0$,~\ref{NP3} and~\ref{P2} can be seen as generalizations of the
conflict constrained problems in the sense that $\ell=0$ (i.e., no
lower bound constraint is imposed), but $u\neq 0$.

Problems \texttt{CCST} and \texttt{CCMST} have been introduced in~\cite{darmann2011paths}, where \texttt{CCMST} is proved to be
polynomial-time solvable if the conflict graph is a matching, and \texttt{CCST} is proved to be
$\NP$-complete if the conflict graph is a forest of paths of length at
most two.
From what is said previously,
we then get that \texttt{G-DCMST}$(G, D, 0, 0, w)$ is polynomial when $D$ is
an oriented matching, and \texttt{G-DCST}$(G,D,0,0)$ is $\NP$-complete
when $D$ is an orientation of a forest of paths of length at most two.
When $D$ is one of these digraphs, observe that $\Delta^-(D)\leq
2$. If $\Delta^-(D) < 2$, the other possibilities for constant
values\footnote{This means that $\ell(e)=k$ and $u(e)=k'\geq k$, for
  all $e\in E(G)$ such that $\dep(e)\neq\emptyset$, and
  $\ell(e)=u(e)=0$ if $\dep(e)=\emptyset$.} of $\ell,u$ are $\ell=0$
and $u=1$, which is trivially polynomial, and
$\ell(e) = u(e) = |\dep(e)|$ for every $e\in E(G)$. For the latter
case, we have results~\ref{NP1} and~\ref{P1}, which leaves open only
the complexity of the optimization problem when $D$ is an oriented
matching. On the other hand, for $\Delta^-(D)=2$, i.e., $D$ contains a
forest of in-stars with at most two leaves, result~\ref{NP1} shows
$\NP$-completeness when $\ell(e) = u(e) = |\dep(e)|$ for every
$e\in E(G)$, the other values of $\ell$ and $u$ remaining open.

The \texttt{CCST} and \texttt{CCMST} problems have also been investigated
in~\cite{zhangKabadiPunnen2011}, where the authors prove that, if the
input graph $G$ is a cactus, then \texttt{CCST} is
polynomial, while \texttt{CCMST} is still
$\NP$-hard. They further show that the optimization problem is
polynomial if the conflict graph $D$ can be turned into a collection
of cliques by the removal of a constant number of vertices, i.e.,
there exists a subset $E'\subseteq E(G) = V(D)$ such that $D-E'$ is a
collection of cliques, and $|E'|$ is bounded by a constant. We prove
something similar here for the generalized problem (result~\ref{P2}).

In~\cite{KLM.13}, the authors investigate a conflict constrained
problem where the conflict graph is only allowed to contain an edge
$ee'$ if $e$ and $e'$ share an endpoint in $G$ (they called
these \emph{forbidden transitions}). Among other results, they prove
that the feasibility problem is $\NP$-complete even if the input graph
$G$ is a complete graph. Practical approaches to the conflict
constrained problem have been presented
in~\cite{samerUrrutia2015,CCPR.18,zhangKabadiPunnen2011}.

Another interesting, recently defined, problem that can be modeled as
a conflict constrained spanning tree problem (and therefore, as a
special case of \texttt{G-DCST}) is the so-called \textsc{Angular
  Constrained Spanning Tree} problem~\cite{AK.17}. In this problem, we
are given a set $V$ of points on the plane, a graph $G=(V,E)$, and an
angle $\alpha$. A spanning tree $T$ is called an
\emph{$\alpha$-spanning tree} if, for every point $v\in V$, there is
an angle on $v$ of size smaller than $\alpha$ containing all the edges
(line segments) of $T$ incident to $v$. Observe that, if we let $D$
contain an arc $(vu,vw)$ whenever the smaller angle formed by $vu$ and
$vw$ is bigger than $\alpha$, then an $\alpha$-spanning tree $T$ also
$(0,0)$-satisfies $D$, and vice-versa. Besides, the conflicts in this
case are forbidden transitions. In \cite{AK.17}, one can find
references on the decision version of the problem, while the
optimization version is investigated in~\cite{CL.19}.

\paragraph{{\it Forcing constraints:}}
Recall that, in the \textsc{Forcing Constrained (Minimum) Spanning Tree}
problem, we are given a pair of graphs $G$ and $D$ such that
$V(D) = E(G)$, and we want to know whether there exists a spanning
tree (find a minimum spanning tree) $T$ of $G$ such that
$E(T)\cap \{u,v\}\neq\emptyset$ for every $uv\in E(D)$. We denote the
feasibility problem by \texttt{FCST} and the optimization problem by
\texttt{FCMST}.

This problem was introduced in~\cite{darmann2011paths}, where the authors
prove that \texttt{FCST} is $\NP$-complete even if the conflict graph
is a forest of paths of length at most two.  To the best of our knowledge, this is the only existing paper that investigates
this problem. Here, we show a reduction
from \texttt{FCST}$(G,D)$ to \texttt{G-DCST}$(G', D', \ell,u)$, where
$\ell(e)\in \{0,1\}$ and $u(e)=|\dep(e)|$ for every $e\in E(G')$, and
the maximum in-degree of $D'$ is~$2$. If weights are being considered, 
such a reduction can be made to preserve the value of the solutions, and
therefore it also applies to the optimization problem.

\paragraph{{\it At least one/all dependency constraints:}}
The following two dependency
constrained problems are introduced in~\cite{vianaCampelo2019}. Given a graph $G$ and a digraph $D$ such that
$V(D) = E(G)$, one wants to know whether there exists a spanning tree
$T$ of $G$ such that: $E(T)\cap \dep(e)\neq \emptyset$ for every $e\in E(T)$
with $\dep(e) \neq \emptyset$, called the \textsc{At Least One
Dependency Constrained Spanning Tree} problem; or $\dep(e)\subseteq E(T)$ for every
$e\in E(T)$, called the \textsc{All Dependency Constrained Spanning Tree} problem. We
denote these problems by \texttt{L-DCST} and \texttt{A-DCST}, and the
related optimization problems by \texttt{L-DCMST} and
\texttt{A-DCMST}, respectively. Note that these are special cases of
our problem.

In~\cite{vianaCampelo2019}, it is proved that both \texttt{L-DCST} and \texttt{A-DCST} are
$\NP$‐complete, even if $G$ is a planar chordal graph with diameter two
or maximum degree three, and $D$ is the disjoint union of arborescences of
height two. Here, we strengthen the constraints on $D$, while also
getting a lower bound on the running time of exponential algorithms
for these problems (result~\ref{NP1}). Observe that this result
comprises cases where the maximum in-degree of $D$ is one, and so the
generalized problem coincides with both \texttt{L-DCST} and
\texttt{A-DCST}.

Still in~\cite{vianaCampelo2019}, the authors prove that \texttt{L-DCMST} and
\texttt{A-DCMST} are $\W[2]$-hard when parameterized by the weight of
a solution, and that, unless $\P=\NP$, they cannot be approximated
with a ratio of $\ln |V(G)|$ even if: $G$ is bipartite; the dependency
relations occur only between adjacent edges of $G$; and each weak
component of $D$ has diameter one. One can notice that the weight of a
solution in their $\W[2]$-hardness reduction is $O(n)$, where
$n=|V(G)|$. This means that there is no \FPT algorithm for
\texttt{L-DCMST} and \texttt{A-DCMST} parameterized by $n$, unless
$\FPT=\W[1]$. This contrasts with the decision problem, which can be
solved in time $O(2^{m}\cdot (n+m)) = O(2^{n^2}\cdot n^2)$, where
$m=|E(G)|$ (result~\ref{P3}).

\paragraph{{\it Maximum degree constraints:}}
Given a graph $G = (V, E)$, and a positive integer $k$, the
\textsc{Maximum Degree Constrained Spanning Tree} problem consists in
deciding whether $G$ has a spanning tree $T$ such that $d_T(v)\le k$
for every $v\in V(G)$, where $d_T(v)$ is the degree of $v$ in $T$. This problem was introduced
in~\cite{deo1968shortest}. Observe that it is $\NP$-complete, even for $k=2$,
since this case generalizes
the Hamiltonian path problem~\cite{garey1979computers}. In
\cite{papadimitriou1984two}, it is proved to be $\NP$-complete even
for grid graphs of maximum degree three.  Also,
\cite{papadimitriou1984two} tackles the Euclidean optimization
version of the problem (i.e., vertices are points on the plane, and
edges are weighted according to the Euclidean distance). The Euclidean
optimization version remains $\NP$-hard when $k \leq 3$, and is polynomial-time solvable when $k \geq 5$, remaining open
for $k=4$.  Several heuristic, approximation, and exact
approaches have been proposed for the problem (see
\cite{Krishnamoorthy2001,Singh2015,Bicalho2016} and references
therein).

Here, we denote the feasibility version of this problem by
\texttt{MDST}, and the optimization version by \texttt{MDMST}. We
present a reduction from \texttt{MDST}$(G,k)$ to
\texttt{G-DCST}$(G',D,0,u)$ where $u(e)\in \{0,k\}$ for every
$e\in E(G')$. The reduction also applies to the optimization problem
since it preserves the value of the solutions.

\paragraph{{\it Minimum degree constraints:}}
Given a graph $G = (V, E)$, and a positive integer $k$, the
\textsc{Minimum Degree Constrained Spanning Tree} problem consists in deciding whether $G$ has a spanning tree $T$
such that $d_T(v)\ge k$ for every non-leaf vertex $v$ of $T$. Here, we
denote the feasibility version of this problem by \texttt{mDST}, and
the optimization version by \texttt{mDMST}.  This problem was
introduced in \cite{de2012min}, where it is shown to be $\NP$-hard for
every $k\in \{4,\cdots, \frac{|V(G)|}{2}\}$.  On the other hand,
\cite{de2012min} proves that the problem can be solved (by inspection)
for degree bounds between $\frac{|V(G)|}{2} + 1$ and $|V(G)| - 1$. In
\cite{de2010md}, the problem was shown to be $\NP$-hard for $k=3$. The case $k\leq 2$ is equivalent to the classical spanning tree problem.
Integer linear programs and solution methods were
proposed in~\cite{AKGUN2010,de2012min,MARTINEZ2014}.

An interesting variation of \texttt{mDST} is obtained when the set of
leaves is fixed in the input. More formally, given a graph $G$, a
subset $C \subseteq V$, and a positive integer $k$, it consists in
finding a spanning tree $T$ of $G$ such that $d_T(v) \geq k$, for
every $v \in C$, and $d_T(v) = 1$, for every $v \in V \setminus C$. We
denote the feasibility version of this problem by \texttt{FmDST}, and
the optimization version by \texttt{FmDMST}.  This problem was
introduced in \cite{dias2017min}, where the authors prove that
\texttt{FmDST} is $\NP$-complete for $k \geq 2$, and
\texttt{FmDMST} is $\NP$-hard even for complete graphs. Also, some
necessary and sufficient conditions are given for feasibility.

Here, we present a reduction from both \texttt{mDST}$(G,k)$ and
\texttt{FmDST}$(G, C, k)$ to \texttt{G-DCST}$(G', D, \ell, u)$ where
$\ell(e)\in \{0,1,k\}$ and $u(e) \in \{1, |\dep(e)|\}$ for every
$e\in E(G)$. Again, our reduction preserves the values of the
solutions and therefore works for the optimization version as well.

\paragraph{{\bf Applications.}}
As $\texttt{G-DCST}$ generalizes all these problems, it inherits their
applications, such as design of wind farm networks
\cite{carrabs2019multiethnic}, VLSI global routing
\cite{robins1995low}, or low-traffic communication networks
\cite{MARTINEZ2014}.  In particular, dependency relations can model
communication systems with protocol conversion restrictions
\cite{Viana16}.  Besides, we can get unified results for all of them
by considering $\texttt{G-DCST}$.

\paragraph{{\bf Organization.}}
In Section~\ref{sec:defs}, we
present the formal definitions and notation used throughout the paper;
in Section~\ref{sec:np} we present our $\NP$-complete results; in
Section~\ref{sec:poly}, our positive results; in
Section~\ref{sec:models}, we show how to model the many constrained
spanning tree problems as special cases of our problem; and in
Section~\ref{sec:conclusion}, we discuss our results and pose some
open questions.

\section{Definitions and notation}\label{sec:defs}

\paragraph{Graphs.} For missing basic definitions on graph theory,
  we refer the reader to~\cite{W96}. Let $G$ be a simple graph (henceforth called simply a graph), and $D$
be a digraph. We denote
by $E(G),E(D)$ the edge set of $G$ and arc set of $D$,
respectively. Also, we denote an edge $\{u,v\}$ of $G$ by $uv$, and
arc with head $v$ and tail $u$ of $D$ by $(u,v)$.  We say that $D$ is
\emph{symmetric} if $(v,u)\in E(D)$ whenever $(u,v)\in E(D)$.  A
(di)graph $G$ ($D$) is \emph{complete} if $uv\in E(G)$
($\{(u,v),(v,u)\}\subseteq E(D)$) for every pair of vertices $u$ and
$v$ in $G$ ($D$). It is \emph{empty} if has no edges (no arcs).

If $C\subseteq V(G)$ is such that $uv\in E(G)$ for every $u,v\in V(G)$, $u\neq v$, then we call $C$ a \emph{clique}. And if there are no edges between vertices in $C$, we say that $C$ is an \emph{independent set}. 
A vertex $v\in V(G)$ is called \emph{universal} if $N(v) = V(G)\setminus \{v\}$, where $N(v)$ stands for the set of neighbors of $v$. A tree $T$ is called a \emph{star} if it has a universal vertex $v$, called \emph{center}. 
Similarly, an \emph{out-star} (\emph{in-star}) is a directed graph $D$
with a vertex $v$ such that any other vertex is an out-neighbor
(in-neighbor) of $v$ and $V(D)\setminus \{v\}$ is an independent set.

\paragraph{Definition of the problems.} Let $G = (V, E)$ be a graph and $D = (E, A)$ be a digraph whose
vertices are the edges of $G$.  We say that $e_1 \in E$ is a
\emph{$D$-dependency} of $e_2 \in E$ if $(e_1, e_2) \in A$.  For
each $e \in E$, we define its \emph{$D$-dependency set} as
$\dep_D(e) = \{e' \in E : (e', e) \in A\}$, and for $E' \subseteq E$,  let
$\dep_D(E') = \cup_{e \in E'} \dep_D(e)$; from now
on we omit $D$ from the subscript whenever it is clear from the
context.  Also, let $\ell, u : E \to \mathbb{N}$ be functions that
assign a non-negative integer to each edge of $G$.  We say that a
subgraph $H$ of $G$ \emph{$(\ell, u)$-satisfies $D$} if
$\ell(e) \leq |\dep(e) \cap E(H)| \leq u(e)$, for every $e \in E(H)$.

We introduce the \textsc{Generalized Dependency Constrained Spanning
  Tree} problem as, given a graph $G$, a digraph $D = (E(G), A)$, and functions $\ell,u:E(G)\rightarrow \mathbb{N}$, deciding whether there exists a spanning tree $T$ of  $G$ such that $T$ $(\ell, u)$-satisfies $D$.  We abbreviate this with \texttt{G-DCST}$(G, D, \ell, u)$. Observe that it
corresponds to the feasibility problem. If we are also given a weight function
$w : E \to \mathbb{R}$, then we define the \textsc{Generalized
  Dependency Constrained Minimum Spanning Tree} problem as the problem
of finding a spanning tree $T^*$ of $G$ that minimizes the weight sum
and that $(\ell, u)$-satisfies $D$; this problem is denoted by
\texttt{G-DCMST}.

\paragraph{Polynomial reductions and Exponential Time Hypothesis.} Given problems $\Pi$ and $\Pi'$, we write $\Pi\preceq_\P \Pi'$ if
there exists a polynomial reduction from $\Pi$ to $\Pi'$. This means
that problem $\Pi'$ is at least as hard as problem $\Pi$.  The
\emph{Exponential Time Hypothesis} (denoted by \ETH) of Impagliazzo et
al.~\cite{ImpagliazzoP01,ImpagliazzoP01-ETH} states that the
\texttt{3-SAT} problem cannot be solved in time $2^{o(n+m)}$, where
$n$ is the number of variables and $m$ the number of clauses of the
input formula. In particular, if it is possible to reduce
\texttt{3-SAT} to problem $\Pi$ and the produced instance has size
linear in the size of the input formula, then the \ETH implies that
problem $\Pi$ cannot be solved in time $2^{|x|}$ either, where $|x|$
denotes the size of the input of $\Pi$. We refer the reader to~\cite{CompComplBook} for basic background on computational complexity. 



\paragraph{Parameterized complexity.} We refer to~\cite{FPTBook} for a recent monograph on parameterized complexity. Here, we recall only some basic definitions.
A \emph{parameterized problem} is a decision problem whose instances are pairs $(x,k) \in \Sigma^* \times \mathbb{N}$, where $k$ is called the \emph{parameter}.
A parameterized problem $L$ is \emph{fixed-parameter tractable} (\FPT) if there exists an algorithm ${\cal A}$, a computable function $f$, and a constant $c$ such that, given an instance $I=(x,k)$ of $L$, we get that 
${\cal A}$ (called an {\sf FPT} \emph{algorithm})  correctly decides whether $I \in L$  in time bounded by $f(k) \cdot |I|^c$. For instance, the \textsc{Vertex Cover} problem parameterized by the size of the solution is \FPT.



Within parameterized problems, the class {\sf W}[1] may be seen as the parameterized equivalent to the class \NP\ of classical optimization problems. Without entering into details (see~\cite{FPTBook} for the formal definitions), a parameterized problem being {\sf W}[1]-\emph{hard} can be seen as a strong evidence that this problem is {\sl not} \FPT. The canonical example of {\sf W}[1]-hard problem is \textsc{Independent Set} parameterized by the size of the solution.
The class {\sf W}[2] of parameterized problems is a class that contains $\W$[1], and such that the problems that are {\sf W}[2]-\emph{hard} are  even more unlikely to be \FPT than those that are {\sf W}[1]-hard (again, see~\cite{FPTBook} for the formal definitions). The canonical example of {\sf W}[2]-hard problem is \textsc{Dominating Set} parameterized by the size of the solution.

\paragraph{Matroids.} We state here some basic tools about matroids that we will use in the algorithms of Section~\ref{sec:poly}, and we refer to~\cite{Oxley-matroids,lawler1976combinatorial} for more background. A (finite) \emph{matroid} $M$ is a pair $(E,\mathcal{I})$, where $E$ is a finite set, called the \emph{ground set}, and $\mathcal{I}$ is a family of subsets of $E$, called the \emph{independent sets}, satisfying the following properties:

\begin{enumerate}
    \item The empty set is independent, that is, $\emptyset \in {\cal I}$.
    \item Every subset of an independent set is independent, that is, for each $A' \subseteq A \subseteq E$, if $A \in {\cal I}$ then $A' \in {\cal I}$. This is called the \emph{hereditary property}.
    \item If $A,B \in {\cal I}$ with $|A|>|B|$, then there exists $x \in A \setminus B$ such that $B \cup \{x\} \in {\cal I}$. This is called the \emph{augmentation property}.
\end{enumerate}

Every graph or multigraph $G=(V,E)$ gives rise to a so-called \emph{graphic matroid} having $E$ as ground set, and a set $F \subseteq E$ is independent if and only if $G[F]$ is acyclic. 

Given a collection ${\cal E} = \{E_1, E_2, \ldots, E_k\}$ of pairwise
  disjoint sets, and integers $\{d_1,\ldots,d_k\}$ such that
  $0\le d_i \leq |E_i|$ for every $i \in [k]= \{1,\ldots,k\}$, the \emph{partition
    matroid} with
  ground set $E = \bigcup_{i=1}^k E_i$ has $S\subseteq E$ as an
  independent set if and only if $|S\cap E_i|\le d_i$ for every
  $i\in [k]$.

The Matroid Intersection Theorem, proved by Edmonds~\cite{edmonds1979matroid}, states that the problem of finding a largest common independent set of two matroids over the same ground set can be solved in polynomial time.


\section{$\NP$-completeness results}\label{sec:np}

In this section, we present our $\NP$-complete results. First, we
impose in Section~\ref{sec:AtleastAll} constraints on the structure of $D$ (also getting constraints
on $G$ as a byproduct), and then we focus in Section~\ref{sec:constraints-functions}
on hardness results imposing constraints on
functions $\ell$ and $u$.

\subsection{Constraints on $D$}\label{sec:AtleastAll}

We prove result~\ref{NP1}. First, we consider the case where $D$ is a forest of
out-stars with at most three vertices. Later we show that the orientation of $D$ in the reduction can be
changed to get a forest of directed paths of length at most two. In addition,  in Theorem~\ref{thr:complexityDCSTinstars}, we modify the reduction to obtain $D$ as a forest of in-stars with at most three vertices, thus closing all the possible orientations of forest of paths of length at most two.

\begin{theorem} \label{thr:complexityDCSTstars}
  \texttt{G-DCST}$(G,D,\ell,u)$ is $\NP$-complete, even when
  $\ell(e) = u(e) = |\dep(e)|$ for every $e\in E(G)$, $D$ is a forest
  of out-stars of maximum degree two, and $G$ is an outerplanar chordal
  graph with diameter at most two. Furthermore, this problem cannot be
  solved in time $2^{o(n+m)}$ unless the \ETH fails, where $n=|V(G)|$ and
  $m=|E(G)|$.
\end{theorem}

\begin{proof}
  Notice that, given a spanning tree $T$, one can check whether $T$
  $(\ell,u)$-satisfies $D$ in polynomial time; hence,
  \texttt{G-DCST}$(G, D,\ell,u)$ is in $\NP$. To prove
  $\NP$-completeness, we present a reduction from \texttt{$(3,2,2)$-SAT}
  to \texttt{G-DCST}.  In the \texttt{$(3,2,2)$-SAT} problem, we are
  given a CNF formula $\phi$ where each clause has at most three
  literals, and each variable appears at most twice positively and at
  most twice negatively. This problem is well-known to be
  $\NP$-complete~\cite{cook1971complexity,garey1979computers}.  So
  consider a CNF formula $\phi$ on $n$ variables and $m$ clauses; we
  build an instance $(G, D, \ell,u)$ of \texttt{G-DCST} as follows
  (follow the construction in Figure \ref{fig:strongerReduction}):

  \begin{itemize}
  \item Add to $G$ vertex $v$, and vertices $v_x$, $v_{\bar x}$, and
    $w_x$ related to each variable $x$, and at most three vertices
    $\{v_c^1,v_c^2,v_c^3\}$ related to each clause $c$ (these vertices
    represent the literals in $c$). Then, make $v$ adjacent to every
    other vertex; $v_x$ adjacent to $v_{\bar x}$ for every variable
    $x$; and add for each clause $c$ a path $(v_c^1,v_c^2)$ or
    $(v_c^1,v_c^2,v_c^3)$, depending on how many literals $c$ has
    (observe that we can suppose that $c$ has at least two
    literals). Edge $vv_x$ will be interpreted as the true assignment of
    $x$, while edge $vv_{\bar x}$ as the false one.

  \item For each variable $x$, add arc $(v_xv_{\bar x}, vw_x)$ to
    $D$. For every variable $x$ and
    each occurrence of $x$ in a clause $c$, say as the $i$-th literal
    in $c$, add to $D$ arc $(vv_x,vv_c^i)$ if $x$ appears positively
    in $c$, or arc $(vv_{\overline{x}},vv_c^i)$ if $x$ appears
    negatively in $c$.  

  \item Finally, let $\ell(e)=u(e)=|\dep(e)|$ for every $e\in E(G)$.
\end{itemize}

\begin{figure}[bht]
  \centering
  \subfigure[Graph $G$.\label{fig:outstarG}]{
    \begin{tikzpicture}
      \tikzstyle{vertex}=[draw,shape=circle,minimum size=20pt,inner sep=0.5pt];

      \node[vertex] (u) at (1, 1) {$v_x$};
      \node[vertex] (v) [below = 1cm of u] {$v$};
      \node[vertex] (w) [above left = 0.4cm and 1cm of v] {$w_x$};

      \node[vertex] (vx) [above right = 0.4cm and 1cm of v] {$v_{\bar x}$};

      \node[vertex] (vc1) [below left = 1cm of v] {$v_c^1$};
      \node[vertex] (vc2) [below = 1cm of v] {$v_c^2$};
      \node[vertex] (vc3) [below right = 1cm of v] {$v_c^3$};

      \draw (u) edge (v);
      \draw (v) edge (w);
      \draw (u) edge (vx);
      \draw (v) edge (vx);

      \draw (v) edge (vc1);
      \draw (v) edge (vc2);
      \draw (v) edge (vc3);
      \draw (vc1) edge (vc2);
      \draw (vc2) edge (vc3);
    \end{tikzpicture}}
  \qquad
  \subfigure[Digraph $D$.\label{fig:outstarD}]{
    \begin{tikzpicture}
      \tikzstyle{vertex}=[draw,shape=circle,minimum size=20pt,inner sep=0.5pt];

      \node[vertex] (a) at (1, 1) {$v_xv_{\bar x}$}; \node[vertex] (al) [below = 1cm of a] {$vw_x$};

      \node[vertex] (ex) [right = 2cm of a] {$vv_x$};
      \node[vertex] (exci) [below left = 1cm of ex] {$vv_{c_1}^1$};
      \node[vertex] (excj) [below right = 1cm of ex] {$vv_{c_2}^1$};

      \node[vertex] (enx) [right = 2cm of ex] {$vv_{\bar x}$};
      \node[vertex] (enxck) [below = 1cm of enx] {$vv_{c_3}^1$};

      \draw (a) edge[->] (al); \draw (ex) edge[->] (exci); \draw (ex) edge[->] (excj);

      \draw (enx) edge[->] (enxck);
    \end{tikzpicture}}
  \caption{Illustration of the reduction from \texttt{$(3, 2, 2)$-SAT} in
    Theorem~\ref{thr:complexityDCSTstars}. In Figure~\ref{fig:outstarG},
    we represent a variable gadget together with a gadget of a clause
    containing three literals. In Figure~\ref{fig:outstarD}, for a variable
    $x$, we represent the dependency between $v_xv_{\bar x}$ and
    $vw_x$, and also the arcs leaving $vv_x$ and $vv_{\bar x}$ when
    $x$ appears positively in $c_1$ and $c_2$, and negatively in
    $c_3$, being related to the first literal in each of these
    clauses.}
  \label{fig:strongerReduction}
\end{figure}
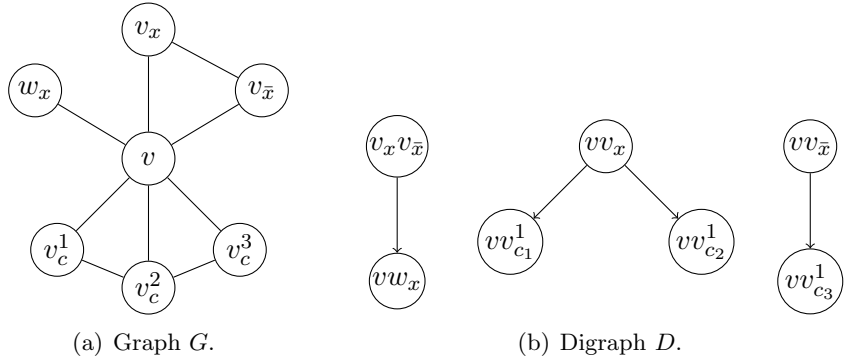

One can see that $G$ is an outerplanar chordal graph, and that each
component (different from the one containing $v_xv_{\overline{x}}$ and
$vw_x$) of $D$ is an out-star from $vv_x$ or $vv_{\bar x}$, for some
variable $x$; we get $\Delta^+(D)\le 2$ by the constraint in the
number of appearances of a literal.

To show the correctness of the reduction, consider first a satisfying
assignment of $\phi$. We build a spanning tree $T$ of $G$ with the following edges: for each
variable $x$, add to $T$ $v_xv_{\bar x}$, $vw_x$, and either $vv_x$ if
$x$ is true, or $vv_{\bar x}$, if $x$ is false; for each clause
$c= (\ell_1 \vee \ell_2 \vee \ell_3)$, add path $(v^1_c,v^2_c,v^3_c)$
and an edge $vv_c^i$ for some $i \in \{1, 2, 3\}$ such that $\ell_i$ is a true literal
(analogously when $c$ has only two literals).  Because, for every
variable $x$, $vw_x$, $v_xv_{\bar x}$ and exactly one edge among
$\{vv_x,vv_{\bar x}\}$ are chosen, and for every clause $c$ exactly
one edge among $\{vv_c^1,vv_c^2,vv_c^3\}$ is chosen, apart from the
path $(v^1_c,v^2_c,v^3_c)$, one can see that $T$ is indeed a spanning
tree of $G$. The dependencies can also be seen to be satisfied since
we only choose an edge $vv_c^i$ if the corresponding literal is true
(hence the dependency is chosen too).

Conversely, let $T$ be a solution for \texttt{G-DCST}$(G, D,\ell, u)$.
For each variable $x$, because $vw_x$ is a cut edge and
$(v_xv_{\bar x},vw_x)\in E(D)$, we get that
$\{vw_x,v_xv_{\bar x}\}\subseteq E(T)$. Besides, for each variable
$x$, since $vv_x$ and $vv_{\bar x}$ form a cut and also form a cycle
with $v_xv_{\bar x}$, we get that exactly one between $vv_x$ and
$vv_{\bar x}$ is in $T$. We then assign $x$ to true if $vv_x\in E(T)$,
and to false otherwise.  Now, consider a clause
$c = (\ell_1 \vee \ell_2 \vee \ell_3)$; since the edges
$\{vv_c^1, vv_c^2, vv_c^3\}$ form a cut, at least one of them is in
$T$, say $vv_c^1\in E(T)$ and say that $x$ is the variable related to
$\ell_1$.  If $\ell_1=x$, then $(vv_x,vv_c^1)\in E(D)$, which implies
that $vv_x\in E(T)$ and that $x$ is true. And if
$\ell_1 = \overline{x}$, then $(vv_{\bar x},vv_c^1)\in E(D)$, which
implies that $vv_{\bar x}\in E(T)$ and that $x$ is false (therefore
$\ell_1$ is true). In any case, $c$ is satisfied. The case where $c$
has only two literals is analogous.

Finally, for the lower bound $2^{o(n+m)}$, just observe that the
constructed instance has size linear in the size of the given formula.
\end{proof}

Observe that when either $\Delta^-(D) = 0$ or $\Delta^+(D)=0$, we get
that $D$ is the empty graph, and that \texttt{G-DCST}$(G, D,\ell,u)$
reduces to deciding whether $G$ is connected. Also, by the previous
theorem, we get that the problem is
$\NP$-complete if $\Delta^-(D) = 1$, thus giving us a dichotomy with regard to the value of
$\Delta^-(D)$. Concerning $\Delta^+(D)$, the previous theorem tells us
that the problem becomes $\NP$-complete for $\Delta^+(D) = 2$. With a
small modification on the previous reduction, we can also get a
dichotomy with regard to $\Delta^+(D)$.

\begin{theorem} \label{thr:complexityGDCSTpaths}
  \texttt{G-DCST}$(G,D,\ell,u)$ is $\NP$-complete, even when
  $\ell(e) = u(e) = |\dep(e)|$ for every $e\in E(G)$, $D$ is a union
  of directed paths with length at most two, and $G$ is a chordal
  outerplanar graph with diameter two.  Furthermore, this problem cannot
  be solved in time $2^{o(n+m)}$ unless the \ETH fails, where
  $n=|V(G)|$ and $m=|E(G)|$.
\end{theorem}

\begin{proof}
  Consider the same construction from the Theorem~\ref{thr:complexityDCSTstars}, except
  that each out-star with two leaves is turned into a directed path of
  length two. Observe that if $(vv_c^i,uv_x,vv_{c'}^j)$ is a path in
  $D$, then the previous arguments might not work simply because we
  might be forced to pick edge $vv_c^i$ when variable $x$ is set to
  true (i.e., edge $uv_x$ is chosen). However, in this case we can
  remove some of the edges of the path $(v^1_c,v^2_c,v^3_c)$ in order
  to avoid cycles. A similar argument is made for out-stars containing
  a vertex of type $vv_x$.
\end{proof}




Recall that \texttt{L-DCST}$(G, D)$ and \texttt{A-DCST}$(G, D)$ denote
the dependency constrained spanning tree problem (\texttt{G-DCST})
where at least one dependency (if any exists) or all dependencies are
satisfied, respectively.  Also, note that, if $\Delta^-(D)\leq 1$,
then we get that \texttt{L-DCST}$(G, D)$ and \texttt{A-DCST}$(G, D)$
coincide with \texttt{G-DCST}$(G, D,\ell, u)$ by assigning
$\ell(e) = u(e) = |\dep(e)|$ for every $e\in E(G)$.  Thus, the
following corollary, which strengthens the results in
\cite{vianaCampelo2019}, is a direct consequence of the previous two
theorems.

\begin{corollary} \label{cor:LA-DCST}
  \texttt{L-DCST}$(G, D)$ and \texttt{A-DCST}$(G, D)$ are
  $\NP$-complete, even if $G$ is an outerplanar chordal graph with
  diameter two, and $D$ is the union of out-stars with
  $\Delta^+(D) = 2$, or the union of paths of lenght at
  most~two. Furthermore, these problems cannot be solved in time
  $2^{o(n+m)}$ unless the \ETH fails, where $n=|V(G)|$ and $m=|E(G)|$.
\end{corollary}

In~\cite{darmann2011paths}, it is shown that \texttt{CCMST} is
$\NP$-complete if the conflict graph is a forest of paths of length at
most two. An orientation of such a forest may lead to directed paths,
out-stars, or in-stars. The $\NP$-completeness of \texttt{G-DCST} in
the first two cases is proved in
Theorems~\ref{thr:complexityDCSTstars} and
\ref{thr:complexityGDCSTpaths}. The case of in-stars is approached
next.

\begin{theorem} \label{thr:complexityDCSTinstars}
  \texttt{G-DCST}$(G, D, \ell, u)$ is \NP-complete, even when
  $\ell(e) = u(e) = |\dep(e)|$ for every $e \in E(G)$, $D$ is a forest
  of in-stars of maximum in-degree two, and $G$ is an outerplanar
  chordal graph with diameter at most two. Furthermore, this problem
  cannot be solved in time $2^{o(n+m)}$ unless the \ETH fails, where
  $n=|V(G)|$ and $m=|E(G)|$.
\end{theorem}

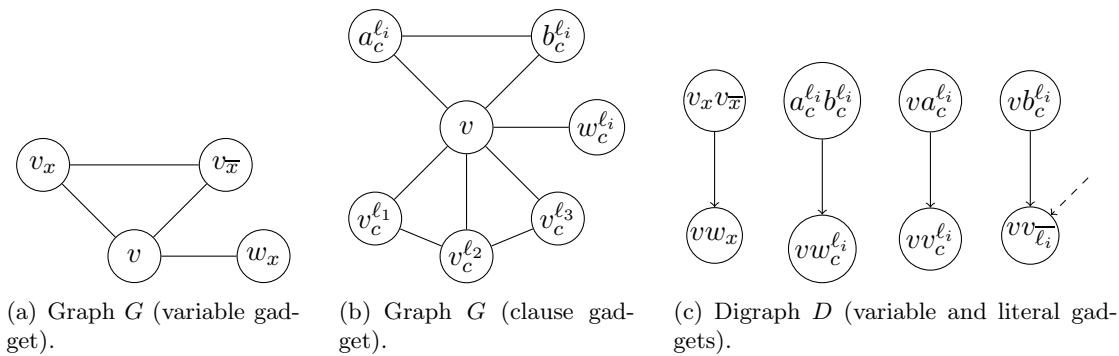
\begin{figure}
  \centering
  \subfigure[Graph $G$ (variable gadget).]{ \label{fig:gadgetVarClauseG}
    \begin{tikzpicture}
      \tikzstyle{vertex}=[draw,shape=circle,minimum size=20pt,inner sep=0.5pt];

      \node[vertex] (v) at (1, 1) {$v$};
      \node[vertex] (vx) [above left = 1cm of v] {$v_x$};
      \node[vertex] (vnx) [above right = 1cm of v] {$v_{\overline{x}}$};
      \node[vertex] (wx) [right = 1cm of v] {$w_x$};


      \draw (v) edge (vx);
      \draw (v) edge (vnx);
      \draw (vx) edge (vnx);
      \draw (v) edge (wx);

    \end{tikzpicture}
  }
  \quad
  \subfigure[Graph $G$ (clause gadget).]{ \label{fig:gadgetLiteralG}
    \begin{tikzpicture}
      \tikzstyle{vertex}=[draw,shape=circle,minimum size=20pt,inner sep=0.5pt];

      \node[vertex] (v) at (1, 1) {$v$};
      \node[vertex] (a) [above left = 1cm of v] {$a_c^{\ell_i}$};
      \node[vertex] (b) [above right = 1cm of v] {$b_c^{\ell_i}$};
      \node[vertex] (wab) [right = 1cm of v] {$w_c^{\ell_i}$};

      \node[vertex] (vc1) [below left = 1cm of v] {$v_c^{\ell_1}$};
      \node[vertex] (vc2) [below = 1cm of v] {$v_c^{\ell_2}$};
      \node[vertex] (vc3) [below right = 1cm of v] {$v_c^{\ell_3}$};

      \draw (v) edge (a);
      \draw (v) edge (b);
      \draw (b) edge (a);
      \draw (v) edge (wab);

      \draw (vc1) edge (vc2);
      \draw (vc3) edge (vc2);
      \draw (v) edge (vc1);
      \draw (v) edge (vc2);
      \draw (v) edge (vc3);
    \end{tikzpicture}
  }
  \quad
  \subfigure[Digraph $D$ (variable and literal gadgets).]{ \label{fig:gadgetVarLiteralD}
    \begin{tikzpicture}
      \tikzstyle{vertex}=[draw,shape=circle,minimum size=20pt,inner sep=0.5pt];

      \node[vertex] (vwx) at (1, 1) {$vw_x$};
      \node[vertex] (vxvnx) [above = 1cm of vwx] {$v_xv_{\overline{x}}$};

      \node[vertex] (ab) [right = 0.5cm of vxvnx] {$a_c^{\ell_i}b_c^{\ell_i}$};
      \node[vertex] (vwab) [below = 1cm of ab] {$vw_c^{\ell_i}$};

      \node[vertex] (va) [right = 0.5cm of ab] {$va_c^{\ell_i}$};
      \node[vertex] (vvc) [below = 1cm of va] {$vv_c^{\ell_i}$};

      \node[vertex] (vb) [right = 0.5cm of va] {$vb_c^{\ell_i}$};
      \node[vertex] (vvl) [below = 1cm of vb] {$vv_{\overline{\ell_i}}$};
      \node (other) [above right = 0.7cm of vvl] {};

      \draw (vxvnx) edge[->] (vwx);
      \draw (ab) edge[->] (vwab);
      \draw (va) edge[->] (vvc);
      \draw (vb) edge[->] (vvl);
      \draw (other) edge[dashed, ->] (vvl);
    \end{tikzpicture}
  }
  \caption{Illustration of the reduction from \texttt{$(3,2,2)$-SAT} used in
    Theorem~\ref{thr:complexityDCSTinstars}.}
  \label{fig:inStar2reduction}
\end{figure}

\begin{proof}
  Notice that a spanning tree $T$ of $G$ can be checked to
  $(l, u)$-satisfy $D$ in polynomial time, thus
  \texttt{G-DCST}$(G, D, \ell, u)$ is in $\NP$.  To prove
  $\NP$-completeness, we again make a reduction from \texttt{$(3, 2, 2)$-SAT}
  to \texttt{G-DCST}.  This way, consider a CNF formula $\phi$ on $n$
  variables and $m$ clauses.  We build an instance $(G, D, \ell, u)$
  of \texttt{G-DCST} as follows (see Figure
  \ref{fig:inStar2reduction}):
  \begin{itemize}
  \item Start by adding a vertex $v$, which will be universal. For
    each variable $x$, add to $G$ vertices $v_x$, $v_{\overline{x}}$,
    and $w_x$, making them adjacent to $v$, and add edge
    $v_xv_{\overline{x}}$; see Figure
    \ref{fig:gadgetVarClauseG}. Selecting edge $vv_x$ will correspond to  the true
    assignment for $x$, and $vv_{\overline{x}}$ to the false one.
  \item For each clause $c = (\ell_1\vee \ell_2\vee \ell_3)$ with three
    literals, add to $G$ vertices
    $\{v_c^{\ell_i},a_c^{\ell_i}, b_c^{\ell_i}, w_c^{\ell_i}\mid i\in
    [3]\}$, and make them adjacent to $v$. Then, add path
    $(v_c^{\ell_1}, v_c^{\ell_2}, v_c^{\ell_3})$, and edges
    $\{a_c^{\ell_i}b_c^{\ell_i}\mid i\in [3]\}$; see Figure
    \ref{fig:gadgetLiteralG}.  Proceed analogously if $c$ has two
    literals.  For each literal $\ell_i$, selecting edge $va_c^{\ell_i}$ will indicate that $c$ is satisfied
    by $\ell_i$, and selecting edge $vb_c^{\ell_i}$ will indicate that $c$ must be
    satisfied by some of its other literals.
  \item For each variable $x$, add arc $(v_xv_{\overline{x}}, vw_x)$
    to $D$. For each clause $c$ and each literal $\ell_i$ of $c$, add
    arcs
    $(a_c^{\ell_i}b_c^{\ell_i}, vw_c^{\ell_i}), (va_c^{\ell_i},
    vv_c^{\ell_i})$, and $(vb_c^{\ell_i}, vv_{\overline{\ell_i}})$ to $D$; see
    Figure \ref{fig:gadgetVarLiteralD}.
  \item Finally, let $\ell(e) = u(e) = |\dep(e)|$ for every
    $e \in E(G)$.
  \end{itemize}

  Observe that $G$ is an outerplanar chordal graph, and that each
  component of $D$ is an in-star. We get $\Delta^-(D) \leq 2$ by the
  constraint in the number of appearances of a literal.

  To show the correctness of the reduction, consider first a
  satisfying assignment of $\phi$.  We build a spanning tree
  $T$ of $G$ as follows. For each variable $x$, choose
  edges $v_xv_{\overline{x}}$ and $vw_x$; then choose edge $vv_x$ if
  $x$ is true, and edge $vv_{\overline{x}}$ otherwise. For each clause
  $c$ with three literals, add edges
  $\{a_c^{\ell_i}b_c^{\ell_i},vw_c^{\ell_i}\mid i \in [3]\}$. Also,
  for each $i\in [3]$, add $\{vv_c^{\ell_i},va_c^{\ell_i}\}$ if
  $\ell_i$ is true; otherwise, add $vb_c^{\ell_i}$. Finally, use path
  $(v_c^{\ell_1}, v_c^{\ell_2},v_c^{\ell_3})$ to connect any possibly
  disconnected vertex. Proceed analogously if $c$ has two literals.
  Denote by $X$ the set of variables of $\phi$, and by $C$ the set of
  clauses; also, write $\ell_i\in c$ to denote the fact that literal
  $\ell_i$ appears in $c$. We first show that $T$ is a spanning tree
  of $G$. It is easy to see that $T$ spans
  $\{v_x,v_{\overline{x}},v,w_x\mid x\in X\}$, and also every
  vertex of degree 1 in $G$. Now, given a clause $c$, because
  each literal $\ell_i$ in $c$ is either true or false, we get that
  either $va_c^{\ell_i}$ or $vb_c^{\ell_i}$ is in $T$, and since
  $a_c^{\ell_i}b_c^{\ell_i}\in E(T)$ we get that $T$ also spans
  $\{a_c^{\ell_i},b_c^{\ell_i}\mid c\in C, \ell_i\in c\}$. Finally,
  for each clause $c$, we know that at least one of its literals is
  true, which means that at least one of the edges linking the path
  $(v_c^{\ell_1}, v_c^{\ell_2},v_c^{\ell_3})$ to $v$ is chosen, and
  since it is always possible to choose edges from this path to
  connect any possible remaining disconnected vertex, we are
  done. Now, we prove that dependencies are satisfied. Dependencies in
  $\{(v_xv_{\overline{x}},vw_x)\mid x\in X\}$, and in
  $\{(a_c^{\ell_i}b_c^{\ell_i},vw_c^{\ell_i})\mid c\in C,\ell_i\in
  c\}$ are all satisfied since all the involved edges are contained in
  $T$. Dependencies in
  $\{(va_c^{\ell_i},vv_c^{\ell_i})\mid c\in C,\ell_i\in c\}$ are also
  valid because we only add these edges together. Finally, given a
  variable $x$, if $x$ is true, then we choose edge $vv_x$, and
  $vb_c^{\ell_i}$ for each clause $c$ such that $\overline{x}$ is the
  $i$-th literal of $c$; and if $x$ is false then we choose edge
  $vv_{\overline{x}}$, and $vb_c^{\ell_i}$ for each clause $c$ such
  that $x$ is the $i$-th literal of $c$. This settles the last type of
  dependencies.

  Conversely, let $T$ be a solution for
  \texttt{G-DCST}$(G, D, \ell, u)$.  For each variable $x$, because
  $vw_x$ is a cut edge and $(v_xv_{\overline{x}},vw_x)\in E(D)$, we
  get that $\{vw_x,v_xv_{\overline{x}}\}\subseteq E(T)$.  Besides, for
  each variable $x$, since $vv_x$ and $vv_{\overline{x}}$ form a cut
  and also form a cycle with $v_xv_{\overline{x}}$, we get that
  exactly one between $vv_x$ and $vv_{\overline{x}}$ is in $T$.  We
  then assign $x$ to true if $vv_x\in E(T)$, and to false otherwise.
  Now, consider a clause $c = (\ell_1 \vee \ell_2 \vee \ell_3$); since
  the edges $\{vv_c^{\ell_1}, vv_c^{\ell_2}, vv_c^{\ell_3}\}$ form a
  cut, at least one of them is in $T$, say $vv_c^{\ell_1}\in
  E(T)$. Hence, because $(va_c^{\ell_1},vv_c^{\ell_1})\in A(D)$, we
  get that $va_c^{\ell_1} \in E(T)$. But then, since
  $a_c^{\ell_1}b_c^{\ell_1}\in E(T)$, we have that
  $vb_c^{\ell_1} \notin E(T)$, which in turn implies that
  $vv_{\overline{\ell}_1} \notin E(T)$ because of the dependency
  $(vb_c^{\ell_1},vv_{\overline{\ell}_1})$. We then conclude that
  $vv_{\ell_1}$ must be chosen, henceforth $\ell_1$ is a true literal
  in $c$.  The case where $c$ has only two literals is analogous.

  Finally, since the constructed instance has size linear in the size
  of $\phi$, we obtain the claimed lower bound $2^{o(n+m)}$ under the \ETH.
\end{proof}

Let us observe that, similarly to Corollary~\ref{cor:LA-DCST}, the results in Theorem~\ref{thr:complexityDCSTinstars} can be stated to \texttt{A-DCST}$(G, D)$, since this problem is equivalent to \texttt{G-DCST}$(G, D, \dep,\dep)$. This also extends the achievements from~\cite{vianaCampelo2019}.

\subsection{Constraints on functions $\ell$ and $u$}
\label{sec:constraints-functions}

In this subsection, we examine the complexity of
\texttt{G-DCST}$(G, D, \ell, u)$ by focusing on functions $\ell$
and $u$. 
Recall that, given constants $c,c'$,
\texttt{G-DCST}$(G, D, c, c')$ denotes the problem restricted to
instances where $\ell(e) = c$ and $u(e)=c'$ for every $e\in
E(G)$. Given a function $f:E\rightarrow \mathbb{N}$ and a positive
integer $c$, we denote by $f+c$ the function obtained from $f$ by
adding $c$ to $f(e)$ for every $e\in E$. The following lemma will be
useful.

\begin{lemma}\label{lem:increasing_ell_u}
  Let $c,\ell, u$ be positive integers with $u\ge \ell$.
  Then, instances $(G,D,\ell,u)$ and $(G',D',\ell+c,u+c)$ of \texttt{G-DCST} are equivalent.
\end{lemma}

\begin{proof}
  Choose any $v\in V(G)$, and let $G'$ be obtained from $G$ by adding
  $\ell+c+1$ vertices of degree one pending in $v$; denote the
  new edges by $e_1,\ldots,e_{\ell+c+1}$. Now, add the symmetric
  clique on vertices $\{e_1,\ldots, e_{\ell+c+1}\}$ to $D$, and add
  $e_ie$ to $D$, for every $i\in [c]$ and every $e\in E(G)$. Let $G'$
  and $D'$ be the obtained graph and digraph, respectively. Finally,
  let $\ell'(e) = \ell+c$ and $u'(e)=u+c$, for every $e\in E(G')$.  We
  prove that $(G,D,\ell,u)$ is a \yes-instance of \texttt{G-DCST}
  if and only if $(G',D',\ell',u')$ is a \yes-instance of
  \texttt{G-DCST}.

  First, let $T$ be a spanning tree of $G$ that $(\ell,u)$-satisfies
  $D$. Let $T'$ be obtained from $T$ by adding
  $e_1,\ldots,e_{\ell+c+1}$. Clearly $T'$ is a spanning tree of $G'$;
  we prove that $T'$ $(\ell',u')$-satisfies $D'$. Let $e\in E(G')$. If
  $e\in E(G)$, because at least $\ell$ and at most $u$ dependencies of
  $e$ are in $T$, and because
  $E(T')\cap \dep_{D'}(e) = (E(T)\cap \dep_D(e))\cup
  \{e_1,\ldots,e_{c}\}$, we get that at least $\ell+c$ and at most
  $u+c$ dependencies of $e$ are in $T'$. And if
  $e\in \{e_1,\ldots,e_{\ell+c+1}\}$, we get that
  $E(T')\cap \dep_{D'}(e) = \{e_1,\ldots,e_{\ell+c+1}\}\setminus \{e\}
  $ and again the constraints hold.

  On the other hand, let $T'$ be a spanning tree of $G'$; we know that
  $\{e_1,\ldots,e_{\ell+ c+1}\}\subseteq E(T')$. Let $T$ be obtained
  from $T'$ by removing these edges, and consider $e\in E(G)$. It follows
  that
  $|\dep_D(e)\cap E(T)| = |(\dep_{D'}(e)\cap E(T'))\setminus
  \{e_1,\ldots,e_{c}\}| = |\dep_{D'}(e)\cap E(T')|-c$, and since
  $\ell+c\le |\dep_{D'}(e)\cap E(T')|\le u+c$ we get that $T$
  $(\ell,u)$-satisfies $D$.
\end{proof}

First, we analyze the cases where $\ell=0$. Recall that in
\texttt{CCST}, whenever an edge $e$ is chosen, no dependencies of $e$
can be chosen; this translates to having $\ell(e) = u(e) = 0$ for
every $e\in E(G)$. So in a sense, when one considers instances
$(G, D, \ell, u)$ of \texttt{G-DCST} where $\ell(e) = 0$ for each
$e \in E(G)$, one can think of the problem as a ``weak" version of
\texttt{CCST}. It thus makes sense to ask whether this version turns
out to be polynomial.  Indeed, we notice that
\texttt{G-DCST}$(G, D, \ell, u)$ with $\ell = 0$ and
$u(e) \geq |\dep_D(e)|$, for each $e \in E(G)$, is an easily solvable
problem since every spanning tree of $G$ trivially
$(\ell, u)$-satisfies $D$. In the following theorem, we see that this
is not the case when $u$ is a constant function.

\begin{theorem}\label{theo:0c}
  Let $c$ be a positive integer. Then \texttt{G-DCST}$(G, D, 0, c)$ is
  $\NP$-complete.
\end{theorem}

\begin{proof}
  Recall that \texttt{CCST}$(G,D)$ is $\NP$-complete \cite{darmann2011paths} and equivalent to
  \texttt{G-DCST}$(G,D',0,0)$ when $D'$ is an arbitrary orientation of $D$.
  Given an instance $(G = (V, E), D)$ of \texttt{CCST}, where $D$ is an undirected graph with $V(D)=E$,
  we construct and equivalent instance $(G', D', 0, c)$ of
  \texttt{G-DCST} as follows (cf.~Figure~\ref{fig:ccmst})
  \begin{itemize}
  \item Let $G'$ be obtained from $G$ by adding a new vertex $p$ and,
    for each $i\in [c]$ and each edge $e\in E(G)$, adding a new vertex
    $p^i_e$. Then, make $p$ adjacent to every $p^i_e$, and to an
    arbitrary vertex $q\in V(G)$. More formally,
    $G' = (V \cup V', E \cup E')$, where
    \[V' = \{p\} \cup \{p_e^i : e \in E, i \in [c]\}\text{ and}\]
    \[E' = \{pq\} \cup \{pp_e^i : e \in E, i \in [c]\}.\]
	
  \item Let $D'$ be obtained from an arbitrary orientation of $D$ by adding an arc $(pp^i_e,e)$ for every $e\in E(G)$ and every $i\in [c]$.
  \end{itemize}
  The constructed instance has size clearly polynomial on the size of
  $(G,D)$ (recall that $c$ is a constant). Note that
  $G'[V'\cup \{q\}]$ is a star.  Now, we show that
  $(G, D)$ is a \yes-instance of \texttt{CCST} if and only if
  $(G', D', 0,c)$ is a \yes-instance of \texttt{G-DCST}. 
	
  First, let $T$ be a solution for \texttt{CCST}$(G, D)$, and let $T'$
  be obtained from $T$ by adding $E'$.  Clearly $T'$ is a spanning
  tree of $G'$; hence it remains to show that $T'$ $(0,c)$-satisfies
  $D'$. For this, consider $e\in E(T')$.  If $e\in E(T)$, then because
  $T$ $(0,0)$-satisfies $D$, we have that
  $E(T')\cap \dep_{D'}(e) = \{pp^1_e,\ldots,pp^c_e\}$ and therefore
  $0\le |E(T')\cap \dep_{D'}(e)|=c$. And if $e\in E'$, we have that
  $\dep_{D'}(e) = \emptyset$ and trivially
  $0 =|\dep_{D'}(e) \cap E(T')| \leq c$.
	
  Conversely, let $T'$ be a spanning tree of $G'$ that
  $(0,c)$-satisfies $D'$, and let $T = T'[V]$.  Because $p$ separates
  $V' \setminus \{p\}$ from $V$,
  we know that $T$ is connected and, therefore, it is a spanning tree
  of $G$; so it remains to show that $T$ $(0,0)$-satisfies $D$. For
  this consider $e\in E(T)$.  Since each edge in $E'$ is a cut edge in
  $G'$, we get that $E' \subseteq E(T')$. Therefore, since
  $|E(T')\cap \dep_{D'}(e)|\le c$ and
  $\{pp^1_c,\ldots,pp^c_e\}\subseteq E(T')\cap \dep_{D'}(e)$, we get
  that $E(T')\cap (E(G)\setminus E') = \emptyset$, i.e.,
  $|E(T)\cap \dep_D(e)| = 0$, as we wanted to show.
\end{proof}

\begin{figure}
  \centering
  \subfigure[$G'$.]{
    \begin{tikzpicture}
      \tikzstyle{vertex}=[draw,shape=circle,minimum size=20pt,inner sep=0.5pt];
		
      \node[vertex] (q) at (1,1) {$q$};
      \node[vertex] (v) [right = 1cm of q] {$v$};
      \node[vertex] (u) [right = 1cm of v] {$u$};
      \node[vertex] (p) [below = 1cm of q] {$p$};
      \node[vertex] (pe11) [below left = 1.5cm of p] {$p_{e_1}^1$};
      \node[vertex] (pe12) [right = 0.3cm of pe11] {$p_{e_1}^c$};
      \node[vertex] (pe22) [below right = 1.5cm of p] {$p_{e_2}^c$};
      \node[vertex] (pe21) [left = 0.3cm of pe22] {$p_{e_2}^1$};

      \draw (pe11) edge[dotted] (pe12);
      \draw (pe21) edge[dotted] (pe22);

      \draw[dotted] (2.5,1) ellipse (3cm and 0.7cm);
      \node (G) at (0,1) {$G$};

      \draw (u) edge node[below] {$e_1$} (v);
      \draw (v) edge node[below] {$e_2$} (q);
      \draw (p) edge (q);
      \draw (p) edge (pe11);
      \draw (p) edge (pe12);
      \draw (p) edge (pe21);
      \draw (p) edge (pe22);
    \end{tikzpicture}}
  \subfigure[$G_c$.]{
    \begin{tikzpicture}
      \tikzstyle{vertex}=[draw,shape=circle,minimum size=20pt,inner sep=0.5pt];

      \node[vertex] (e1) at (1,1) {$e_1$};
      \node[vertex] (e2) [right = 1cm of e1] {$e_2$};

      \draw (e1) edge (e2);
    \end{tikzpicture}}
  \subfigure[$D$.]{
    \begin{tikzpicture}
      \tikzstyle{vertex}=[draw,shape=circle,minimum size=20pt,inner sep=0.5pt];

      \node[vertex] (e1) at (1,1) {$e_1$};
      \node[vertex] (e2) [right = 1cm of e1] {$e_2$};
      \node[vertex] (ae1) [above left = 1cm of e1] {$pp_{e_1}^1$};
      \node[vertex] (ae2) [above right = 1cm of e1] {$pp_{e_1}^c$};
      \node[vertex] (af1) [below left = 1cm of e2] {$pp_{e_2}^1$};
      \node[vertex] (af2) [below right = 1cm of e2] {$pp_{e_2}^c$};
      \draw[dotted] (2.2,-0.2) -- (3.2,-0.2);
      \draw[dotted] (0.5, 2.2) -- (1.5, 2.2);

      \draw (e1) edge[bend left, ->] (e2);
      \draw (e2) edge[bend left, ->] (e1);
      \draw (ae1) edge[->] (e1);
      \draw (ae2) edge[->] (e1);
      \draw (af1) edge[->] (e2);
      \draw (af2) edge[->] (e2);
    \end{tikzpicture}}
  \caption{Illustration of the reduction for \texttt{CCMST}.}
  \label{fig:ccmst}
\end{figure}
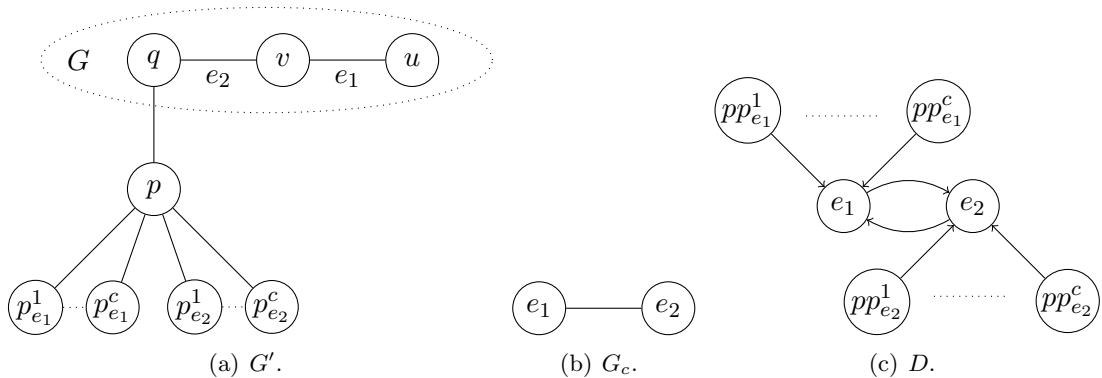

Combining Lemma~\ref{lem:increasing_ell_u} and Theorem~\ref{theo:0c}, we get result~\ref{NP2}, that is, 
\texttt{G-DCST}$(G, D, \ell, u)$ is $\NP$-complete for every
combination of constant values $\ell$ and $u$.

\begin{corollary}\label{cor:hard-constant}
  For every pair of positive integers $\ell,u$ with $\ell\le u$, we
  have that \texttt{G-DCST}$(G, D, \ell, u)$ is $\NP$-complete.
\end{corollary}

\begin{proof}
  Let $c = u-\ell$. By Theorem~\ref{theo:0c}, we have that
  \texttt{G-DCST}$(G, D, 0, c)$ is $\NP$-complete, and by
  Lemma~\ref{lem:increasing_ell_u}, we have that
  \texttt{G-DCST}$(G, D, \ell, c+\ell = u)$ also is.
\end{proof}

As we have already mentioned, if $u(e) = |\dep(e)|$ for every
$e\in E(G)$, then \texttt{G-DCST}$(G, D, 0, u)$ is easy since any
spanning tree $(0,u)$-satisfies $D$. Hence, it is natural to ask whether the
problem continues to be easy when $u(e)$ is just slightly smaller than
$|\dep(e)|$. The following corollary is trivially obtained from previous results. It answers the aforementioned question negatively and yields result~\ref{NP3}.
Given a positive integer $c$,
we denote by $\dep-c$ the function $u:E(G)\rightarrow \mathbb{N}$
defined by $u(e) = \max\{|\dep(e)|-c,0\}$.

\begin{corollary}\label{cor:dep-1}
  \texttt{G-DCST}$(G, D, 0, \dep-1)$ is $\NP$-complete, even when $D$
  is a collection of directed paths of length at most two.
\end{corollary}

\begin{proof}
  Assume that $D$ is a collection of directed paths of length at
  most two.  Then, $\Delta^-(D)=1$ and therefore $\dep-1$ is equal to
  the constant function zero. It means that we are considering
  \texttt{G-DCST}$(G, D, 0, 0)$. This problem is equivalent to
  \texttt{CCST}$(G,D')$, where $D'$ is the (undirected) underlying
  graph of $D$. Then $D'$ is a collection of paths of length at most
  two, and it is known that \texttt{CCST}$(G,D')$ is
  $\NP$-complete~\cite{darmann2011paths}.
\end{proof}

\section{Positive results}\label{sec:poly}

In this section, we present some polynomial cases of the
\texttt{G-DCST} problem. Whenever possible, we deal with the
optimization version instead, i.e., we consider that we are also given
a weight function $w$ and that the objective is to find a spanning
tree that $(\ell,u)$-satisfies $D$ having minimum weighted
sum-weight. This is denoted by \texttt{G-DCMST}$(G,D,\ell,u,w)$.

Similarly to the previous section, we first focus in Section~\ref{sec:algo-D} on constraints on the dependency graph $D$, and then in
Section~\ref{sec:poly_ell_u} on constraints on the functions $\ell$ and $u$. Finally, we present in Section~\ref{sec:exponential_algorithm} a simple exponential-time algorithm to solve the problem.

\subsection{Constraints on $D$}
\label{sec:algo-D}

We start by investigating the case where $D$ is a collection of
directed paths of length at most one (note that there might be some edges of $G$ that are isolated vertices in $D$).  So, given a graph $G = (V, E)$ and a digraph $D = (E, A)$,
write $E$ as $\{e_1, \dots, e_m\}$ and assume that
$A = \{(e_i, e_{i + t}) : i \in \{1, \dots, t\} \}$, for some $t \leq \lfloor \frac{m}{2} \rfloor$.
Let $S = \{e_i : i \in \{t + 1, \dots, 2t\}\}$, i.e., $S$ is the set
of the edges $e\in E(G)$ with $\dep(e)\neq \emptyset$. The idea is to
find a subset of edges $S' \subseteq S$ that connects the components
of $G-S$ and such that the dependencies of $S'$ do not form a cycle in
$G$.  Let $\mathcal{C} = \{C_1, \dots, C_k\}$ be the set of connected
components of $G-S$, and let $H$ be the multigraph obtained from $G$
by contracting each $C_i$ to a single vertex (loops are
removed). Observe that there is an injection from the
multiedges of $H$ to the edges in $S$.  Hence, given $S'\subseteq S$,
we can pick $H(S')$ to be the subgraph of $H$ induced by $S'$, i.e.,
$H(S')$ has ${\cal C}$ as vertex set, and there is an edge between
$C_i,C_j$, $i\neq j$, for each edge $e\in S'$ with one endpoint in
$C_i$ and the other one in $C_j$. Observe that some of the edges of
$S$ might appear inside a component. However, as we will see, these
edges can actually be ignored, and this is why we do not need to add
the loops in $H$.

\begin{lemma} \label{lem:matchingIdea}
  Let $G = (V, E)$ be a graph, $D = (E, A)$ be a digraph, 
  before, 
  and $\ell,u$ be such that
  $\ell(e) = u(e) = |\dep(e)|\in \{0,1\}$ for every $e\in E(G)$.  If
  there exists $S' \subseteq S$ such that $H(S')$ is a spanning tree
  of $H$ and $(V, \dep(S'))$ is acyclic, then
  \texttt{G-DCST}$(G, D,\ell, u)$ is feasible.  Conversely, any
  feasible solution of \texttt{G-DCST}$(G, D,\ell,u)$ contains such a  subset $S'$.
\end{lemma}

\begin{proof}
First, consider $S' \subseteq S$ such that $H(S')$ is a spanning tree of $H$ and $(V, \dep(S'))$ is acyclic. Observe that, because $(V, \dep(S'))$ is acyclic, $S' \subseteq S$, and $S \cap \dep(S') = \emptyset$, we get that each $C_i$
  is a connected component of $G-S$. Thus, we can add edges of $G-S$ to
  $(V, \dep(S'))$ so as to obtain a spanning forest $F$ of $G$ having
  connected components with the same vertex sets as
  $C_1,\ldots,C_k$. After this, just add edges of $S'$ to $F$; because
  $H(S')$ is a spanning tree of $H$, we get that the obtained graph
  $T$ connects all components of $F$ without forming a cycle (i.e.,
  $T$ is a spanning tree of $G$). Finally, since
  $\dep(S')\subseteq E(T)$ and $S\cap E(T) = S'$, we get that $T$
  satisfies $D$.

  Conversely, let $T = (V, E_T)$ be a feasible solution of
  \texttt{DCST}$(G, D)$.  As $T$ is a spanning tree of $G$, we get
  that the edges in $E_T\cap S$ must connect the components of $G-S$,
  i.e., $H(E_T\cap S)$ is connected. Thus, choose $S'$ as the edge set
  of any spanning tree of $H(E_T\cap S)$. We get that $S'$ also forms
  a spanning tree of $H$, and since $S'\subseteq E_T$ and $T$
  satisfies $D$, we get that $\dep(S')\subseteq E_T$ and therefore $(V, \dep(S'))$
  cannot contain a cycle.
\end{proof}

In the following theorem we use the Matroid Intersection Theorem~\cite{edmonds1979matroid} to get result~\ref{P1}.


\begin{theorem}
  Let $G = (V, E)$ be a graph, $D = (E, A)$ be a digraph such that
  each component is a directed path of length at most $1$, and
  $\ell,u$ be such that $\ell(e) = u(e) = |\dep(e)|\in \{0,1\}$ for
  every $e\in E(G)$. Then \texttt{DCST}$(G,D)$ can be solved in
  polynomial time.
\end{theorem}

\begin{proof}
  Given $E'\subseteq E(G)$, denote by $G(E')$ the subgraph
  $(V(G), E')$.  Let $H$ be obtained as before. Also, let
  $\mathcal{I}_1 = \{S' \subseteq S \mid G(\dep(S')) \mbox{ is
    acyclic} \}$ and
  $\mathcal{I}_2 = \{S' \subseteq S \mid H(S') \mbox{ is acyclic} \}$.
  We have that $(S, \mathcal{I}_1)$ and $(S, \mathcal{I}_2)$ are
  matroids (on a common ground set $S$), since they are equivalent to
  the graphic matroids of the graph $(V, \dep(S))$ and of the multigraph $H$,
  respectively.  According to Lemma \ref{lem:matchingIdea},
  \texttt{DCST}$(G, D)$ is feasible if and only if there is
  $S' \in \mathcal{I}_1 \cap \mathcal{I}_2$ such that $|S'| = k - 1$,
  where $k$ is the number of components of $G' = (V, E \setminus S)$.
  The existence of such $S'$ can be checked in polynomial time using
  the Matroid Intersection Theorem \cite{edmonds1979matroid}.
\end{proof}

\subsection{Constraints on $\ell$ and $u$}\label{sec:poly_ell_u}

Recall that \texttt{CCMST}$(G,G_c,w)$ is equivalent to
\texttt{G-DCMST}$(G,D,0,0,w)$ when $D$ is a symmetric digraph.  From a
result in~\cite{zhangKabadiPunnen2011} for the \texttt{CCMST} problem,
we get that \texttt{G-DCMST}$(G, D, 0,0, w)$ is solvable in polynomial
time when $D$ is the union of complete digraphs.  We generalize this
result for upper bound functions that have the same value on each
clique (result~\ref{P2}).

\begin{theorem}
  Let $(G, D, 0, u, w)$ be an instance of \texttt{G-DCMST} such that
  $D = D_1 \cup D_2 \cup \dots \cup D_k$ is the union of $k$ disjoint
  complete digraphs and, for each $i \in [k]$, there exists
  $u_i \in [|V(D_i)|]$ such that $u(e) = u_i - 1$ for every
  $e \in V(D_i)$. Then, \texttt{G-DCMST}$(G, D, 0, u, w)$ can be
  solved in polynomial time.
\end{theorem}

\begin{proof}
  Note that, in this case, a solution for
  \texttt{G-DCMST}$(G, D, 0, u, w)$ can have at most $u_i$ edges from
  $D_i$, for each $i \in [k]$.  We show that solving such an instance can be formulated as a matroid intersection problem.
  

  
  Observe that a subgraph
  $T$ of $G$ $(0,u)$-satisfies $D$ if and only if
  $|E(T)\cap D_i|\le u_i$ for every $i\in [k]$. Therefore, if $M$ is
  the partition matroid on $E(G)$ defined by
  $\{V(D_1),\ldots,V(D_k)\}$ and $\{u_1,\ldots,u_k\}$, we get that
  $T\subseteq G$ is a solution for our problem if and only if $T$ is a
  spanning tree of $G$ and $E(T)$ is an independent set of
  $M$. Hence, if $M'$ is the graphic matroid associated with $G$ (where the independent sets are the sets of edges inducing a spanning tree of $G$), we can  solve our problem in polynomial time by applying  the Matroid Intersection Theorem~\cite{edmonds1979matroid} to the matroids $M$ and $M'$. 
\end{proof}

Observe that the same approach used in the previous theorem does not
work when the values of $u$ can differ inside the same clique. For
instance, suppose that $D_1$ is the complete digraph on edges
$\{e_1,\ldots,e_4\}$ and that $u(e_1)=1$ and $u(e_i)=2$ for every
$i\in \{2,3,4\}$. Then $S_1 =\{e_1,e_2\}$ and $S_2=\{e_2,e_3,e_4\}$
are acceptable subsets within a solution, however because of $e_1$,
there does not exist $e\in S_2\setminus S_1$ such that $S_1\cup \{e\}$
is an acceptable subset. This means that the \emph{augmentation
  property} (cf. Section~\ref{sec:defs}) is not satisfied and the subsets that define feasible
solutions do {\sl not} form a matroid.

\subsection{Exponential exact algorithm}\label{sec:exponential_algorithm}

In this section, we present an exponential exact algorithm for
\texttt{G-DCST}$(G,D,\ell,u)$, as states result~\ref{P3}. Recall that, as a consequence of Theorem 3 in~\cite{vianaCampelo2019}
for the optimization problems \texttt{L-DCMST} and \texttt{A-DCMST}, we get that the {\sl optimization} problem \texttt{G-DCMST}$(G,D,\ell,u, w)$ is $\W[2]$-hard when parameterized by $n=|V(G)|$. 
The importance of the algorithm below, despite its simplicity, is
that it separates the complexity of the feasibility and the optimization problems. 

\begin{theorem}
  \texttt{G-DCST}$(G,D,\ell,u)$ can be solved in time
  $O(2^{m}\cdot (n+m))$, where $n=|V(G)|$ and $m=|E(G)|$.
\end{theorem}

\begin{proof}
  It suffices to observe that, given a subset $S\subseteq E(G)$, one
  can test in time $O(n+m)$ whether $S$ forms a spanning tree of $G$,
  and whether the constraints imposed by $D,\ell,u$ are satisfied by
  $S$. Because there are $2^{m}$ possible subsets to be tested, the
  theorem follows.
\end{proof}

\section{Modeling other constrained problems}	 \label{sec:models}

As discussed in the introduction, the spanning tree problem has
been investigated under the most various constraints. In this section,
we show that some of them can be modeled as special cases of our
problem. We have already remarked that this is the case for the
\textsc{Conflict Constrained Spanning Tree} problem. One should also notice that, in all
the reductions presented below, when the feasibility version
reduces to our problem, so does the optimization version.
It is a matter of observing that the graph in the source instance is a subgraph of the graph in the \texttt{G-DCST} instance, and that the reduction preserves the solution value when keeping the weights of the initial edges and setting to zero the weights of the new edges.

In Subsection \ref{subsec:forcingModeling}, we present a reduction from \textsc{Forcing
  CST} (denoted by \texttt{FCST}), in Subsection
\ref{subsec:maxdegModeling}, a reduction from \textsc{Maximum Degree
  CST} (denoted by \texttt{MDST}), and in Section~\ref{subsec:mindegModeling}, from \textsc{Minimum Degree
  CST} and \textsc{Fixed-Leaves Minimum Degree CST} (denoted by \texttt{mDST} and \texttt{FmDST}, respectively).

\subsection{Forcing constrained spanning trees} \label{subsec:forcingModeling}

Recall that, given graphs $G$ and $D$ such that $V(D) = E(G)$,
\texttt{FCST} consists in deciding whether $G$ has a spanning tree $T$
such that $|E(T)\cap \{e,f\}|\ge 1$ for every $ef\in E(D)$.

\begin{theorem}
Denote by $\texttt{G-DCST}^*$ the problem $\texttt{G-DCST}$ restricted to instances $(G', D', \ell, 2)$ such that $\ell(e)\in \{0,1\}$
  for every $e\in E(G')$ and $\Delta^-(D')=2$. Then, $\texttt{FCST}\preceq_\P \texttt{G-DCST}^*$.
\end{theorem}

\begin{proof}
  Let $(G,D)$ be an instance of \texttt{FCST}, and construct $G',D'$
  as follows (cf. Figure~\ref{ForcingMST_reduction}). Choose any
  $v\in V(G)$, and let $G'$ be obtained from $G$ by adding, for each
  $ee'\in E(D)$, a pendant degree one vertex
  in $v$; denote the new edge by
  $p_{ee'}$. Then, let $D'$ be the digraph with vertex set $E(G')$ and
  arcs $(e,p_{ee'})$ and $(e',p_{ee'})$ for every $ee'\in
  E(D)$. Finally, let $\ell(e) = 0$ if $e\in E(G)$, and $\ell(e)=1$
  otherwise. We prove that $(G, D)$ is a \yes-instance of
  \texttt{FCST} if and only $(G, D, \ell, 2)$ is a \yes-instance of
  \texttt{G-DCST}.

  First, let $T$ be a solution for \texttt{FCST}, and let $T'$ be
  obtained from $T$ by adding $p_{ee'}$ for every $ee'\in
  E(D)$. Clearly $T'$ is a spanning tree of $G'$, and since
  $|E(T)\cap \{e,e'\}|\ge 1$ for every $ee'\in E(D)$, we get that
  $|E(T')\cap \dep_{D'}(p_{ee')}|\ge 1$. The reverse implication can be proved similarly.
\end{proof}

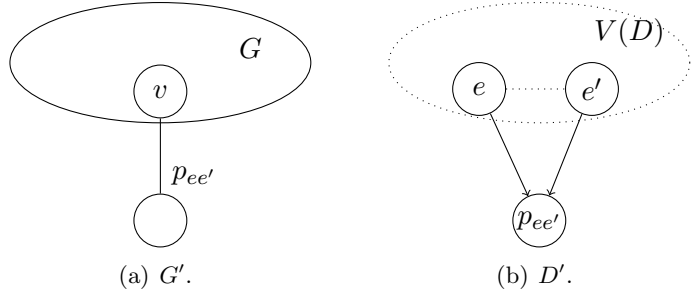
\begin{figure}[bht]
  \centering
  \subfigure[$G'$.]{
    \begin{tikzpicture}
      \tikzstyle{vertex}=[draw,shape=circle,minimum size=20pt,inner sep=0.5pt];

      \node[vertex] (ee) at (1, 1) {};
      \node[vertex] (v) [above = 1cm of ee] {$v$};
      \draw (1,3.1) ellipse (2cm and 0.8cm);
      \node (G) at (2.2,3.3) {$G$};
      \draw (ee) edge (v);
      \node (pee) at (1.45,1.6) {$p_{ee'}$};
    \end{tikzpicture}}
  \qquad
  \subfigure[$D'$.]{
    \begin{tikzpicture}
      \tikzstyle{vertex}=[draw,shape=circle,minimum size=20pt,inner sep=0.5pt];

      \node[vertex] (pee) at (1, 1) {$p_{ee'}$};
      \node[vertex] (e) at (0.2,2.75) {$e$};
      \node[vertex] (e2) at (1.7,2.75) {$e'$};
      \draw[dotted] (1,3.1) ellipse (2cm and 0.8cm);
      \node (D) at (2.2,3.5) {$V(D)$};
      \draw (e) edge[->] (pee);
      \draw (e2) edge[->] (pee);
      \draw[dotted] (e) to (e2);
    \end{tikzpicture}}

  \caption{Illustration of the reduction for \texttt{FCST}.}
  \label{ForcingMST_reduction}
\end{figure} 

\subsection{Max-degree constrained spanning trees} \label{subsec:maxdegModeling}

Given a graph $G = (V, E)$ and a positive integer $k$, recall that in
the \texttt{MDST}$(G, k)$ problem we want to find a spanning tree $T$
such that $d_T(v) \leq k$, for every $v \in V(G)$. We prove that a
generalized version of this problem reduces to ours. In
$\texttt{MDST}(G, d^*)$, instead of being given an integer $k$, we are
given a function $d^*:V(G)\rightarrow \mathbb{N}$ that separately sets
upper bounds to the degrees of the vertices.

\begin{theorem}
Denote by $\texttt{G-DCST}^{**}$ the problem $\texttt{G-DCST}$ restricted to instances $(G, D, 0, u)$. Then, $\texttt{MDST}\preceq_\P \texttt{G-DCST}^{**}$.
\end{theorem}

\begin{proof}
  Let $(G, d^*)$ be an instance of \texttt{MDST}.  We build an
  equivalent instance $(G', D, 0, u)$ of \texttt{G-DCST} as follows (cf. Figure
  \ref{MaxDegMST_reduction}).
  \begin{itemize}
  \item $G' = (V \cup V', E \cup E')$, where
    $V' = \{v' \mid v \in V\}$ has a copy of each vertex, and
    \mbox{$E' = \{vv' \mid v \in V\}$} connects each vertex $v\in V$
    to its copy $v'\in V'$.
  \item $D = (E \cup E', A)$, where
    $A = \{(uv, uu'), (uv, vv') \mid uv \in E\}$.
  \item $u(e) = 0$, for each $e \in E$, while $u(vv') = d^*(v)$, for
    each $v \in V$.
  \end{itemize}

  Observe that $\dep(vv')$ is the set of edges incident to $v$ in $G$
  for every $v\in V$, and that $\dep(e) = \emptyset$ for every
  $e\in E$. Also, note that because $\ell(e) = u(e) = 0 = |\dep(e)|$
  for every $e\in E$, and $\ell(e) = 0$ for every $e\in E'$, the only
  real constraints being imposed by $D$ are upper bounds on the chosen
  dependencies for edges in $E'$. More specifically, we get that a
  spanning tree $T'$ of $G'$ $(\ell,u)$-satisfies $D$ if and only if
  $|E(T')\cap \dep(vv')|\le d^*(v)$ for every $vv'\in E'\cap
  E(T')$. Note that each $v'$ has degree one
  in $G'$, which implies that every edge in $E'$ must be part of every
  spanning tree of $G'$. Thus, we get that $(V,E_T)\subseteq G$ is a
  tree that satisfies the maximum degree constraints if and only if
  $(V\cup V', E_T\cup E')$ $(\ell,u)$-satisfies $D$.\textbf{}
\end{proof}

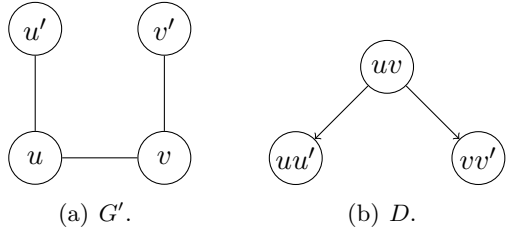
\begin{figure}[bht]
  \centering
  \subfigure[$G'$.]{
    \begin{tikzpicture}
      \tikzstyle{vertex}=[draw,shape=circle,minimum size=20pt,inner sep=0.5pt];

      \node[vertex] (u) at (1, 1) {$u$};
      \node[vertex] (v) [right = 1cm of u] {$v$};
      \node[vertex] (u') [above = 1cm of u] {$u'$};
      \node[vertex] (v') [above = 1cm of v] {$v'$};

      \draw (u) edge (v);
      \draw (u) edge (u');
      \draw (v) edge (v');
    \end{tikzpicture}}
  \qquad
  \subfigure[$D$.]{
    \begin{tikzpicture}
      \tikzstyle{vertex}=[draw,shape=circle,minimum size=20pt,inner sep=0.5pt];

      \node[vertex] (uv) at (1, 1) {$uv$};
      \node[vertex] (eu) [below left = 1cm of uv] {$uu'$};
      \node[vertex] (ev) [below right = 1cm of uv] {$vv'$};

      \draw (uv) edge[->] (eu);
      \draw (uv) edge[->] (ev);
    \end{tikzpicture}}

  \caption{Illustration of the reduction for \texttt{MDST}.}
  \label{MaxDegMST_reduction}
\end{figure} 

\subsection{Minimum degree constrained spanning trees} \label{subsec:mindegModeling}

Given a graph $G$ and a positive integer $k$, recall that the
\texttt{mDST}$(G, k)$ problem consists in finding a spanning tree $T$
of $G$ such that $d_T(v)\ge k$ for every nonleaf vertex $v\in V(T)$.
We introduce a generalized version of this problem, denoted by
\texttt{G-mDST}$(G, \ell, u)$, where we replace the integer $k$ by
functions $\ell,u: V\to \mathbb{N}$ and require that each nonleaf
vertex $v$ of $T$ satisfies $\ell(v)\leq d_T(v) \leq u(v)$. Clearly,
\texttt{mDST}$(G, k)$ is a special case of
\texttt{G-mDST}$(G, \ell, u)$ where $\ell(v)=k$ and $u(v)=d(v)$ for
every $v\in V(G)$.

\begin{theorem} \label{MD-MST_particular_G-DCMST} 
  $\texttt{G-mDST}\preceq_\P \texttt{G-DCST}$.
\end{theorem}
\begin{proof}
  Given an instance $(G = (V, E), \ell, u)$ of \texttt{G-mDST}, we
  build an instance $(G', D, \ell', u')$ of \texttt{G-DCST} as follows
  (cf. Figure~\ref{MD-MST_reduction}):
  \begin{itemize}
  \item $G' = (V \cup V', E \cup E')$, where
    $V' = \{v_1, v_2, v_3 \mid v \in V\}$ and
    $E' = \{v v_1, vv_2, v_1v_3, v_2v_3 \mid v \in V\}$;
  \item $D = (E \cup E', A_1\cup A_2)$, where
    $A_1 = \{(u v, vv_1), (uv, vv_2) \mid uv\in E\}$ and
    $A_2 = \{(v_2v_3, v_1v_3), (v_1v_3, v_2v_3) \mid v \in V\}$;
  \item For each $e \in E$, let $\ell'(e) = u'(e) = 0$; and for each
    $v \in V$, let $\ell'(vv_1) = \ell(v)$, $u'(vv_1) = u(v)$, and
    $\ell'(e) = u'(e) = 1$, for each $e\in \{vv_2,v_1v_3,v_2v_3\}$.
  \end{itemize}

  Observe that, for each $v\in V$ and $i\in \{1,2\}$, we have that
  $\dep(vv_i)$ is the set of edges incident to $v$ in $G$. We show that
  $(G, \ell, u)$ is a \yes-instance of \texttt{G-mDST} if and only
  $(G', D, \ell', u')$ is a \yes-instance of \texttt{G-DCST}.

\begin{figure}
  \centering
  \subfigure[$G'$.]{
    \begin{tikzpicture}
      \tikzstyle{vertex}=[draw,shape=circle,minimum size=20pt,inner sep=0.5pt];

      \node[vertex] (u) at (1, 1) {$u$};
      \node[vertex] (v) [right = 3cm of u] {$v$};

      \node[vertex] (v1) [above left = 1cm of v] {$v_1$};
      \node[vertex] (v2) [above right = 1cm of v] {$v_2$};
      \node[vertex] (v3) [above = 1cm of v] {$v_3$};

      \node[vertex] (u1) [above left = 1cm of u] {$u_1$};
      \node[vertex] (u2) [above right = 1cm of u] {$u_2$};
      \node[vertex] (u3) [above = 1cm of u] {$u_3$};

      \draw (u) edge (v);


      \draw (u) edge (u1);
      \draw (u) edge (u2);
      \draw (u1) edge (u3);
      \draw (u2) edge (u3);

      
      \draw (v) edge (v1);
      \draw (v) edge (v2);
      \draw (v1) edge (v3);
      \draw (v2) edge (v3);
    \end{tikzpicture}}
  \quad
  \subfigure[$D$.]{
    \begin{tikzpicture}
      \tikzstyle{vertex}=[draw,shape=circle,minimum size=20pt,inner sep=0.5pt];

      \node[vertex] (uv) at (1, 1) {$uv$};
      \node[vertex] (e1u) [left = 1cm of uv] {$uu_1$};
      \node[vertex] (e2u) [below left = 1cm of uv] {$uu_2$};
      \node[vertex] (e1v) [right = 1cm of uv] {$vv_1$};
      \node[vertex] (e2v) [below right = 1cm of uv] {$vv_2$};

      \node[vertex] (e3u) [right = 0.5cm of e1v] {$u_1u_3$};
      \node[vertex] (e4u) [right = 1cm of e3u] {$u_2u_3$};

      \node[vertex] (e3v) [below = 0.5cm of e3u] {$v_1v_3$};
      \node[vertex] (e4v) [right = 1cm of e3v] {$v_2v_3$};

      \draw (uv) edge[->] (e1v);
      \draw (uv) edge[->] (e2v);
      \draw (uv) edge[->] (e1u);
      \draw (uv) edge[->] (e2u);

      \draw (e3u) edge[->, bend left = 30] (e4u);
      \draw (e4u) edge[->, bend left = 30] (e3u);

      \draw (e3v) edge[->, bend left = 30] (e4v);
      \draw (e4v) edge[->, bend left = 30] (e3v);
    \end{tikzpicture}}

  \caption{Illustration of the reduction for \texttt{G-mDST}.}
  \label{MD-MST_reduction}
\end{figure}
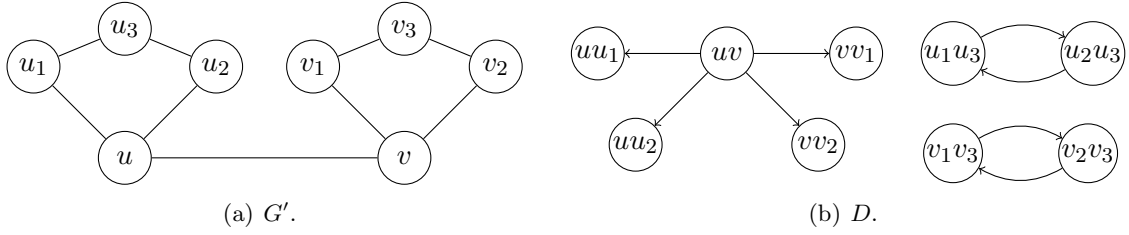

  First, let $T = (V, E_T)$ be a solution for
  \texttt{G-mDST}$(G, \ell, u)$.  We expand $T$ into
  $T' = (V \cup V', E_{T'}) \subseteq G'$, where $E_{T'}$ is equal to
  $E_T$ together with the following edges. For each $v \in V$, add
  $v_1v_3$ and $v_2v_3$, and if $v$ is a leaf in $T$ then add $vv_2$,
  otherwise add $vv_1$. Observe that $T'$ is a spanning tree of $G'$
  such that $T = T'[V]$.  It remains to show that the $D$-dependencies
  are satisfied.  Every edge $e \in E_T$ has $\dep(e) = \emptyset$ and
  $\ell'(e) = u'(e) = 0$, so
  $\ell'(e) \leq |\dep(e) \cap E(T')| \leq u'(e)$ trivially follows.
  The remaining types of edges in $E_{T'}\setminus E_T$ are analyzed below:
  \begin{enumerate}
  \item[(i)]$v_iv_3$, for $v \in V$ and $i\in \{1,2\}$: recall that
    $v_1v_3$ is the unique dependency of $v_2v_3$, and vice-versa, and
    they are both in $T'$. Hence
    $1=\ell'(v_iv) \leq |\dep(v_iv) \cap E(T')| \leq u'(v_iv)=1$.
  \item[(ii)] $vv_2$: by construction of $T'$, we get that $v$ is necessarily a leaf in
    $T$. Since exactly one edge of $\dep(vv_2)$ is in $E_{T'}$, namely
    the edge incident to $v$ in $T$, we have that
    $1=\ell'(vv_2) \leq |\dep(vv_2) \cap E(T')| \leq u'(vv_2)=1$.
  \item[(iii)] $vv_1$: by construction of $T'$, we get that $v$ in $T$ is not
    a leaf in $T$.  From the feasibility of $T$, we know that
    $\ell(v)\leq d_T(v) \leq u(v)$, i.e., at least $\ell(v)$ and at
    most $u(v)$ edges of $\dep(e_1^v)$ are in $E_T \subseteq E_{T'}$,
    which implies that
    $\ell(v) = \ell'(vv_1) \leq |\dep(vv_1) \cap E(T')| \leq u'(vv_1) =
    u(v)$.
  \end{enumerate}

  Conversely, suppose that $T' = (V \cup V', E_{T'})$ is a solution
  for \texttt{G-DCST}$(G', D, \ell', u')$.  We show that $T=T'[V]$ is
  a solution for \texttt{G-mDST}$(G, \ell, u)$.  Due to the dependency
  constraints, we get that $v_1v_3$ and $v_2v_3$ are in $T'$, for each
  $v \in V$.  From this, and since $vv_1$ and $vv_2$ are a cut in
  $G'$, exactly one of $vv_1$ and $vv_2$ is in $T'$, for each
  $v \in V$.  Take $v \in V$.  If $vv_1$ is in $T'$, then there are
  at least $\ell(v)$ and at most $u(v)$ edges incident to $v$ in $T'$.  And if
  $vv_2$ is in $T'$, then there is exactly one edge $uv \in E$ in
  $E_{T'}$.  Therefore, either $\ell(v)\leq d_T(v) \leq u(v)$ or
  $d_T(v) = 1$. Since $T'$ is a spanning tree of $G'$, $T$ is a
  spanning tree of $G$, and thus $T$ is a solution of
  \texttt{GD-MST}$(G, \ell, u, w)$.
\end{proof}

Finally, given a graph $G$, a subset $C \subseteq V$, and a function
$\ell: C \to \mathbb{Z}^+$, recall that \texttt{FmDST}$(G, C, \ell)$
denotes the problem of finding a spanning tree $T$ of $G$ such that
$d_T(u) \geq \ell(u)$ for every $u \in C$, and $d_T(v) = 1$ for
every $v \in V \setminus C$ (i.e., the set of leaves is prefixed).
Observe that the same reduction of Theorem~\ref{MD-MST_particular_G-DCMST} can be applied to this problem
by removing edge $e_2^v$ for each $v \in C$, and edge $e_1^v$ for each
$v \in V \setminus C$. We then get the following:

\begin{theorem}
  $\texttt{FmDST}\preceq_\P \texttt{G-DCST}$.
\end{theorem}

\section{Conclusion}	\label{sec:conclusion}

In this paper, we investigated a dependency constrained spanning tree
problem that generalizes many previously studied spanning tree problems with local constraints, as for instance degree constraints. We then inherit
all of the $\NP$-completeness results for these problems, as well as
polynomial results and practical approaches to our problem will
therefore hold for these other problems. Interestingly, there are other spanning tree problems that impose global
constraints on the tree, as for instance, a bound on the diameter of
the produced
tree~\cite{camerini1980complexity,camerini1983complexity}, or on the
number of leaves~\cite{garey1979computers,lu1992power}.  A good
question therefore is whether problems with this kind of constraints
can be modeled within our framework.

\begin{question}
  Can instances of CST problems with global constraints be modeled as
  \texttt{G-DCST} instances?
\end{question}

Concerning $\NP$-completeness results, we have investigated the
problem under restrictions on the dependency digraph $D$, and on the
lower and upper bound functions $\ell$ and $u$. Among other
restrictions, we have proved that \texttt{G-DCST} is $\NP$-complete
when $D$ is either a forest of directed paths, a forest of out-stars,
or a forest of in-stars, and each component has at most three vertices. In
the first two cases, we have considered all possible values for
constant functions $\ell,u$. The following cases for
in-stars remain open.

\begin{question}
  What is the complexity of \texttt{G-DCST}$(G,D,\ell,u)$ when $D$ is
  a forest of in-stars with at most three vertices and $(\ell,u)=(0,1)$
  or, for every $e\in E(G)$ such that $\dep(e)\neq \emptyset$,
  $\ell(e)=1$ and $u(e)\in \{1,2\}$?
\end{question}

Concerning positive results, we have proved that
\texttt{G-DCST}$(G,D,\dep,\dep)$ can be solved in polynomial time when
$D$ is an oriented matching. We ask whether this also holds for the
optimization problem (recall that \texttt{G-DCMST}$(G,D,0,0)$ is
polynomial in this case~\cite{DPS.09}).

\begin{question}
  What is the complexity of \texttt{G-DCMST}$(G,D,\dep,\dep)$ when $D$
  is an oriented matching?
\end{question}

Finally, we have proved that \texttt{G-DCMST} can be solved in
polynomial time when $\ell=0$, $D$ is a collection of symmetric graphs
$D_1,\ldots,D_k$, and $u(e)=u(e')$ whenever $e,e'$ are within the same
component, and that \texttt{G-DCST}$(G,D,\ell,u)$ can be solved in
time $O(2^{m}\cdot (n+m))$ by a naive brute-force algorithm, where $n=|V(G)|$ and $m=|E(G)|$. The
latter result is important in the face of the fact that, as a
byproduct of a result in~\cite{vianaCampelo2019}, we get that no
algorithm running in time $2^{O(n)}$ exists for the optimization
problem, unless \ETH fails. Also, the results presented in Section~\ref{sec:AtleastAll}
imply that no algorithm that runs in time $2^{o(n+m)}$ exists for the
feasibility problem under the \ETH, which means that the algorithm presented in
Section~\ref{sec:exponential_algorithm} is asymptotically optimal. We
mention that our algorithm can also be seen as an FPT algorithm
parameterized by $m$. We ask whether the problem is FPT under other
parameters.

\begin{question}
  Under which parameters is \texttt{G-DCST} or \texttt{G-DCMST} \FPT?
\end{question}

In order to identify parameters for the above question, note that the maximum degree of the input graph $G$ is {\sl not} enough, since a particular case of the problem is already \NP-complete restricted to graphs with maximum degree at most three~\cite{vianaCampelo2019}. Similarly, the maximum degree of the dependency graph $D$ is {\sl not} enough either, as another particular case of the problem is $\NP$-complete even if $D$
is a forest of paths of length at most two (see e.g.~\cite{darmann2011paths}, as well as our results presented in Section~\ref{sec:AtleastAll}).

A promising candidate parameter for obtaining \FPT algorithms is the \emph{treewidth} of the input graph $G$ (see~\cite{FPTBook} for the definition); note that the treewidth of the underlying graph of $D$ is {\sl not} enough by the last sentence of the above paragraph, since forests have treewidth~one. Suppose that, in order to use Courcelle's Theorem~\cite{Courcelle90} or any of its optimization variants~\cite{ArnborgLS91}, one tries to express the \texttt{G-DCST} problem in \emph{monadic second-order} ({\sf MSO}) logic. In order to guarantee that the dependencies of $D$ are satisfied for every edge $e$ of the desired spanning tree of the input graph $G$, one would probably need to {\sl evaluate} the functions $\ell(e)$ and $u(e)$ {\sl inside} the eventual {\sf MSO} formula, and this seems to be a fundamental hurdle since these values are a priori unrelated to the treewidth of $G$. Nevertheless, for the particular case of \texttt{G-DCST} (or \texttt{G-DCMST}) where both functions $\ell$ and $u$ are {\sl constants} (or even equal to some constant value that depends on the treewidth of $G$), it is indeed possible, using standard techniques, to write such an {\sf MSO} formula expressing the problem, and therefore it is \FPT parameterized by the treewidth of the input graph. Note that this restriction of the \texttt{G-DCST} problem is \NP-complete by Corollary~\ref{cor:hard-constant}, for every pair of positive integers $\ell,u$ with $\ell\le u$. The next natural step would be to consider as parameters {\sl both} the treewidth of $G$ and the maximum degree of $D$.

\bibliographystyle{plain}
\bibliography{paper}{}

\end{document}